\documentclass[12pt]{iopart}

\usepackage[pdftex]{graphicx}
\usepackage{graphicx}
\usepackage{mathrsfs}
\usepackage{amssymb}
\expandafter\let\csname equation*\endcsname\relax
\expandafter\let\csname endequation*\endcsname\relax
\usepackage{amsmath}
\usepackage{amsthm}
\usepackage{bm}
\usepackage{hyperref}
\usepackage{float}
\usepackage[utf8x]{inputenc}
\usepackage{amsfonts}
\usepackage{amsbsy}
\usepackage{bbold}
\usepackage{bbm}
\usepackage{mathdots}

\usepackage{iopams}

\newtheorem{proposition}{Proposition}
\newtheorem{example}{Example}

\newtheorem{corollary}{Corollary}

\newcommand{\bra}[1]{\langle #1|}
\newcommand{\ket}[1]{| #1 \rangle }
\newcommand{\ip}[2]{{\langle #1|}{ #2 \rangle }}

\begin{document}

    \title{Absolutely separating quantum maps and channels}

    \author{S N Filippov$^{1,2}$, K Yu Magadov$^1$ and M A Jivulescu$^3$ }

    \address{$^1$ Moscow Institute of Physics and Technology,
        Institutskii Per. 9, Dolgoprudny, Moscow Region 141700, Russia}

    \address{$^2$ Institute of Physics and Technology of the Russian
        Academy of Sciences, Nakhimovskii Pr. 34, Moscow 117218, Russia}

    \address{$^3$ Department of Mathematics, Politehnica University of
        Timi\c{s}oara, Victoriei Square 2, 300006 Timi\c{s}oara, Romania}

    \ead{sergey.filippov@phystech.edu}

\begin{abstract}
Absolutely separable states $\varrho$ remain separable under
arbitrary unitary transformations $U \varrho U^{\dag}$. By example
of a three qubit system we show that in multipartite scenario
neither full separability implies bipartite absolute separability
nor the reverse statement holds. The main goal of the paper is to
analyze quantum maps resulting in absolutely separable output
states. Such absolutely separating maps affect the states in a
way, when no Hamiltonian dynamics can make them entangled
afterwards. We study general properties of absolutely separating
maps and channels with respect to bipartitions and multipartitions
and show that absolutely separating maps are not necessarily
entanglement breaking. We examine stability of absolutely
separating maps under tensor product and show that $\Phi^{\otimes
N}$ is absolutely separating for any $N$ if and only if $\Phi$ is
the tracing map. Particular results are obtained for families of
local unital multiqubit channels, global generalized Pauli
channels, and combination of identity, transposition, and tracing
maps acting on states of arbitrary dimension. We also study the
interplay between local and global noise components in absolutely
separating bipartite depolarizing maps and discuss the input
states with high resistance to absolute separability.
\end{abstract}

\section{Introduction}

The phenomenon of quantum entanglement is used in a variety of
quantum information
applications~\cite{horodecki-2009,nielsen-2000}. The distinction
between entangled and separable states has an operational meaning
in terms of local operations and classical communication, which
cannot create entanglement from a separable quantum
state~\cite{heinosaari-ziman-2012}. Natural methods of
entanglement creation include interaction between subsystems,
measurement in the basis of entangled states, entanglement
swapping~\cite{yurke-1992,zukowski-1998,pan-1998}, and dissipative
dynamics towards an entangled ground
state~\cite{reiter-2012,lin-2013}. On the other hand, dynamics of
any quantum system is open due to inevitable interaction between
the system and its environment. The general transformation of the
system density operator for time $t$ is given by a dynamical map
$\Phi_t$, which is completely positive and trace preserving (CPT)
provided the initial state of the system and environment is
factorized~\cite{breuer-2002}. CPT maps are called quantum
channels~\cite{holevo-2012}. Dissipative and decoherent quantum
channels describe noises acting on a system state. Properties of
quantum channels with respect to their action on entanglement are
reviewed in the
papers~\cite{aolita-2015,filippov-2014,filippov-ziman-2013}.

Suppose a quantum channel $\Phi$ such that its output
$\varrho_{\rm out} = \Phi[\varrho]$ is separable for some initial
system state $\varrho$. It may happen either due to entanglement
annihilation of the initially entangled state
$\varrho$~\cite{moravcikova-ziman-2010,filippov-rybar-ziman-2012},
or due to the fact that the initial state $\varrho$ was separable
and $\Phi$ preserved its separability. Though the state
$\varrho_{\rm out}$ is inapplicable for entanglement-based quantum
protocols, there is often a possibility to make it entangled by
applying appropriate control operations, e.g. by activating the
interaction Hamiltonian $H$ among constituting parts of the system
for a period $\tau$. It results in a unitary transformation
$\varrho_{\rm out} \rightarrow U \varrho_{\rm out} U^{\dag}$,
where $U=\exp(-\rmi H\tau/\hbar)$, $\hbar$ is the Planck constant.
Thus, if a quantum system in question is controlled artificially,
one can construct an interaction such that the state $U
\varrho_{\rm out} U^{\dag}$ may become entangled. It always takes
place for pure output states $\varrho_{\rm out} = \ket{\psi_{\rm
out}}\bra{\psi_{\rm out}}$, however, such an approach may fail for
mixed states. These are \emph{absolutely separable} states that
remain separable under action of \emph{any} unitary operator
$U$~\cite{Kus,Verstraete}. Properties of absolutely separable
states are reviewed in the
papers~\cite{Hildebrand,Johnston,Jivulescu,arunachalam-2015}. Even
if the dynamical map $\Phi$ is such that $\Phi[\varrho]$ is
absolutely separable for a given initial state $\varrho$, one may
try and possibly find a different input state $\varrho'$ such that
$\Phi[\varrho']$ is not absolutely separable, and the system
entanglement could be recovered by a proper unitary
transformation. It may happen, however, that whatever initial
state $\varrho$ is used, the output $\Phi[\varrho]$ is always
absolutely separable. Thus, a dynamical map $\Phi$ may exhibit an
\emph{absolutely separating} property, which means that its output
is always absolutely separable and cannot be transformed into an
entangled state by any Hamiltonian dynamics. The only
deterministic way to create entanglement in a system acted upon by
the absolutely separating channel $\Phi$ is to use a nonunitary
CPT dynamics afterwards, e.g. a Markovian dissipative process
$\widetilde{\Phi}_t =\rme^{t{\cal L}}$ with the only fixed point
$\varrho_{\infty}$, which is entangled. From experimental
viewpoint it means that absolutely separating noises should be
treated in a completely different way in order to maintain
entanglement.

The goal of this paper is to characterize absolutely separating
maps $\Phi$, explore their general properties, and illustrate
particular properties for specific families of quantum channels.

The paper is organized as follows.

In section~\ref{section-as-states}, we review properties of
absolutely separable states and known criteria for their
characterization. We establish an upper bound on purity of
absolutely separable states. Also, we pay attention to the
difference between absolute separability with respect to a
bipartition and that with respect to a multipartition. We show the
relation between various types of absolute separability and
conventional separability in tripartite systems. In
section~\ref{section-as-channels}, we review general properties of
absolutely separating maps and provide sufficient and (separately)
necessary conditions of absolutely separating property.
Section~\ref{section-n-tensor-stable-as} is devoted to the
analysis of $N$-tensor-stable absolutely separating maps, i.e.
maps $\Phi$ such that $\Phi^{\otimes N}$ is absolutely separating
with respect to any valid bipartition. In
section~\ref{section-specific-as-channels}, we consider specific
families of quantum maps, namely, local depolarization of qubits
(section~\ref{subsection-depolarizing}), local unital maps on
qubits (section~\ref{subsection-Pauli}), generalized Pauli
diagonal channels constant on axes~\cite{NatRus}
(section~\ref{subsection-generalized-Pauli}). In
section~\ref{subsection-ctit}, we consider a combination of
tracing map, transposition, and identity transformation acting on
a system of arbitrary dimension. Such maps represent a
two-parametric family comprising a global depolarization channel
and the Werner-Holevo channel~\cite{werner-holevo-2002} as partial
cases. In section~\ref{subsection-bipartite-depolarizing}, we deal
with the recently introduced three-parametric family of bipartite
depolarizing maps~\cite{lami-huber-2016}, which describe a
combined physical action of local and global depolarizing noises
on a system of arbitrary dimension. In
section~\ref{section-discussion}, we discuss the obtained results
and focus attention on initial states $\varrho$ such that
$\Phi_t[\varrho]$ remains not absolutely separable for the maximal
time $t$. In section~\ref{section-conclusions}, brief conclusions
are given.

\section{Absolutely separable states}
\label{section-as-states}

Associating a quantum system with the Hilbert space $\mathcal{H}$,
a quantum state is identified with the density operator $\varrho$
acting on $\mathcal{H}$ (Hermitian positive semidefinite operator
with unit trace). By $\mathcal{S}(\mathcal{H})$ denote the set of
quantum states. We will consider finite dimensional spaces
$\mathcal{H}_d$, where the subscript $d$ denotes ${\rm
dim}\mathcal{H}$.

\subsection{Bipartite states}

A quantum state $\varrho \in \mathcal{S}(\mathcal{H}_{m n})$, $m,n
\geqslant 2$ is called separable with respect to the particular
partition $\mathcal{H}_{mn} = \mathcal{H}_{m}^{A} \otimes
\mathcal{H}_{n}^{B}$ on subsystems $A$ and $B$ if $\varrho$ adopts
the convex sum resolution $\varrho = \sum_k p_k \varrho_k^A
\otimes \varrho_k^B$,  $p_k \geqslant 0$, $\sum_k p_k \geqslant
0$~\cite{werner-1989}. We will use a concise notation
$\mathcal{S}(\mathcal{H}_{m}^{A} | \mathcal{H}_{n}^{B})$ for the
set of such separable states. Usually, subsystems $A$ and $B$
denote different particles or modes~\cite{amosov-filippov-2017},
depending on experimentally accessible degrees of freedom. If the
state $\varrho \notin \mathcal{S}(\mathcal{H}_{m}^{A} |
\mathcal{H}_{n}^{B})$, then $\varrho$ is called entangled with
respect to $\mathcal{H}_{m}^{A} | \mathcal{H}_{n}^{B}$.

In contrast, a quantum state $\varrho \in
\mathcal{S}(\mathcal{H}_{m n})$ is called \emph{absolutely
separable} with respect to partition $m|n$ if $\varrho$ remains
separable with respect to \emph{any} partition $\mathcal{H}_{mn} =
\mathcal{H}_{m}^{A} \otimes \mathcal{H}_{n}^{B}$ without regard to
the choice of $A$ and $B$~\cite{Hildebrand,Johnston,Jivulescu}.
Denoting by $\mathcal{A}(m|n)$ the set of absolutely separable
states, we conclude that $\mathcal{A}(m|n) = \cap_{A,B}
\mathcal{S}(\mathcal{H}_{m}^{A} | \mathcal{H}_{n}^{B})$. The
physical meaning of absolutely separable states is that they
remain separable without respect to the particular choice of
relevant degrees of freedom. In terms of the fixed partition
$\mathcal{H}_{m}^{A} | \mathcal{H}_{n}^{B}$, the state $\varrho
\in \mathcal{S}(\mathcal{H}_{m n})$ is absolutely separable with
respect to $m|n$ if and only if $U \varrho U^{\dag} \in
\mathcal{S}(\mathcal{H}_{m}^{A} | \mathcal{H}_{n}^{B})$ for any
unitary operator $U$.

Let us notice that specification of partition is important. For
instance, 12-dimensional Hilbert space allows different partitions
$2|6$ and $3|4$. Moreover, one can always imbed the density
operator $\varrho \in \mathcal{S}(\mathcal{H}_{d_1})$ into a
higher-dimensional space $\mathcal{S}(\mathcal{H}_{d_2})$, $d_2 >
d_1$ and consider separability with respect to other partitions.

Absolutely separable states cannot be transformed into entangled
ones by unitary operations. In quantum computation circuits, the
application of unitary entangling gates (like controlled-NOT or
$\sqrt{\text{SWAP}}$) to absolutely separable states is useless
from the viewpoint of creating entanglement. Thus, a quantum
dynamics transforming absolutely separable states into entangled
ones must be non-unitary. Example of a dynamical map $\Phi$, which
always results in an entangled output, is a Markovian process
$\Phi_t=\rme^{t{\cal L}}$ with the only fixed point
$\varrho_{\infty} \not\in
\mathcal{S}(\mathcal{H}_m|\mathcal{H}_n)$; the output state
$\Phi_t[\varrho]$ is always entangled for some time $t > t_{\ast}$
even if the input state $\varrho$ is absolutely separable.

In analogy with the absolutely separable states, absolutely
classical spin states were introduced
recently~\cite{bohnet-waldraff-2017}.
The paper~\cite{bohnet-waldraff-2017} partially answers the question:
what are the states of a spin-$j$ particle that remain classical
no matter what unitary evolution is applied to them? These states
are characterized in terms of a maximum distance from the
maximally mixed spin-$j$ state such that any state closer to the
fully mixed state is guaranteed to be classical.

\subsection{Criteria of absolute separability with respect to bipartition}

Note that two states $\varrho$ and $V \varrho V^{\dag}$, where $V$
is unitary, are both either absolutely separable or not. In other
words, they exhibit the same properties with respect to absolute
separability. Let $V$ diagonalize $V \varrho V^{\dag}$, i.e. $V
\varrho V^{\dag} = {\rm diag} (\lambda_1, \ldots, \lambda_{mn})$,
where $\lambda_1, \ldots, \lambda_{mn}$ are eigenvalues of
$\varrho$. It means that the property of absolute separability is
defined by the state spectrum only.

A necessary condition of separability is positivity under partial
transpose (PPT)~\cite{peres-1996,horodecki-1996}: $\varrho \in
\mathcal{S}(\mathcal{H}_{m}^{A} | \mathcal{H}_{n}^{B})
\Longrightarrow \varrho^{\Gamma_B} = \sum_{i,j=1}^n I^A \otimes
\ket{j}^B \bra{i} \, \varrho \, I^A \otimes \ket{j}^B \bra{i}
\geqslant 0$, where $I$ is the indentity operator,
$\{\ket{i}\}_{i=1}^n$ is an orthonormal basis in
$\mathcal{H}_n^B$, and positivity of Hermitian operator $O$ means
$\bra{\varphi} O \ket{\varphi} \geqslant 0$ for all
$\ket{\varphi}$. In analogy with absolutely separable states one
can introduce \emph{absolutely PPT} states with respect to
partitioning $m|n$, namely, $\varrho \in
\mathcal{S}(\mathcal{H}_{m n})$ is absolutely PPT with respect to
$m|n$ if $\varrho^{\Gamma_B} \geqslant 0$ for all decompositions
$\mathcal{H} = \mathcal{H}_m^A \otimes
\mathcal{H}_n^B$~\cite{zyczkowski-1998,Hildebrand}. Equivalently,
$\varrho \in \mathcal{S}(\mathcal{H}_{m n})$ is absolutely PPT
with respect to $m|n$ if $(U\varrho U^{\dag})^{\Gamma_B} \geqslant
0$ for all unitary operators $U$. The set of absolutely PPT states
with respect to $m|n$ denote $\mathcal{A}_{\rm PPT}(m|n)$. It is
clear that $\mathcal{A}(m|n) \subset \mathcal{A}_{\rm PPT}(m|n)$
for all $m,n$.

The set $\mathcal{A}_{\rm PPT}(m|n)$ is fully characterized in
\cite{Hildebrand}, where necessary and sufficient conditions on
the spectrum of $\varrho$ are found under which $\varrho$ is
absolutely PPT. These conditions become particularly simple in the
case $m=2$: the state $\varrho \in \mathcal{S}(\mathcal{H}_{2n})$
is absolutely PPT if and only if its eigenvalues $\lambda_1,
\ldots, \lambda_{2n}$ (in decreasing order $\lambda_1 \geqslant
\ldots \geqslant \lambda_{2n}$) satisfy the following inequality:
\begin{equation}
\label{2-n} \lambda_1 \leqslant \lambda_{2n-1} + 2
\sqrt{\lambda_{2n} \lambda_{2n-2}}.
\end{equation}

Due to the fact that separability is equivalent to PPT for
partitions $2|2$ and $2|3$~\cite{horodecki-1996},
$\mathcal{A}(2|2) = \mathcal{A}_{\rm PPT}(2|2)$ and
$\mathcal{A}(2|3) = \mathcal{A}_{\rm PPT}(2|3)$. Moreover, the
recent study~\cite{Johnston} shows that $\mathcal{A}(2|n) =
\mathcal{A}_{\rm PPT}(2|n)$ for all $n=2,3,4,\ldots$. Thus,
equation~\eqref{2-n} is a necessary and sufficient criterion for
absolute separability of the state $\varrho \in
\mathcal{S}(\mathcal{H}_{2n})$ with respect to partition $2|n$.

For general $m,n$ there exists a sufficient condition of absolute
separability based on the fact that the states $\varrho$ with
sufficiently low purity ${\rm tr}[\varrho^2]$ are
separable~\cite{zyczkowski-1998,Gurvits,Gurvits2,Hildebrand-2007}.
Suppose the state $\varrho \in \mathcal{S}(\mathcal{H}_{mn})$
satisfies the requirement
\begin{equation}
\label{ball} {\rm tr}[\varrho^2] = \sum_{k=1}^{mn} \lambda_k^2
\leqslant \frac{1}{mn-1},
\end{equation}

\noindent then $\varrho \in \mathcal{S}(\mathcal{H}_{m} |
\mathcal{H}_{n})$. Since unitary rotations $\varrho \rightarrow
U\varrho U^{\dag}$ do not change the Frobenius norm, the states
inside the separable ball \eqref{ball} are all absolutely
separable, i.e. \eqref{ball}$\Longrightarrow \varrho \in
\mathcal{A}(m|n)$.

Suppose $\varrho\in\mathcal{A}_{\rm PPT}(m|n)$, then decreasingly ordered eigenvalues
$\lambda_1,\ldots,\lambda_{mn}$ of $\varrho$
satisfy~(\cite{Jivulescu}, theorem 8.1)
\begin{equation}
\label{necessary} \lambda_1 \leqslant \lambda_{mn-2} +
\lambda_{mn-1} + \lambda_{mn}.
\end{equation}

\noindent Since $\mathcal{A}(m|n) \subset
\mathcal{A}_{\rm PPT}(m |n)$,
equation~\eqref{necessary} represents a readily computable necessary
condition of absolute separability with respect to bipartition
$m|n$. The physical meaning of equation~\eqref{necessary} is that the
absolutely separable state cannot be close enough to any pure
state, because for pure states $\lambda_1^{\downarrow} = 1$ and
$\lambda_2^{\downarrow} = \ldots = \lambda_{mn}^{\downarrow} = 0$,
which violates the requirement \eqref{necessary}.

Moreover, a factorized state $\varrho_1 \otimes \varrho_2$ with
$\varrho_1 \in \mathcal{S}(\mathcal{H}_m)$ and $\varrho_2 \in
\mathcal{S}(\mathcal{H}_n)$, $m,n \geqslant 2$, cannot be
absolutely separable with respect to partition $m|n$ if either
$\varrho_1$ or $\varrho_2$ belongs to a boundary of the state
space. In fact, a boundary density operator $\varrho_1 \in
\partial\mathcal{S}(\mathcal{H}_m)$ has at least one zero
eigenvalue, which implies at least $n \geqslant 2$ zero
eigenvalues of the operator $\varrho_1 \otimes \varrho_2$.
Consequently, $\lambda_{(m-1)n} = \ldots = \lambda_{mn} = 0$ and
equation~\eqref{necessary} cannot be satisfied.

Condition \eqref{necessary} imposes a limitation on the purity of
absolutely separable states. Next proposition provides a
quantitative description of the maximal ball, where all absolutely
separable states are located.

\begin{proposition}
\label{proposition-not-AS-purity} An absolutely separable state
$\varrho \in \mathcal{A}(m|n)$ necessarily satisfies the
inequality
\begin{equation}
\label{max-ball-exact} 1 + \sqrt{\frac{k \, {\rm tr}[\varrho^2] -
1}{k-1}} \leqslant 3 k \sqrt{\frac{{\rm tr}[\varrho^2]}{mn+8}}
\text{~~~if~~~} \frac{1}{k} \leqslant {\rm tr}[\varrho^2]
\leqslant \frac{1}{k-1}, \ k=2,3,\ldots,mn
\end{equation}

\noindent and its simpler implication
\begin{equation}
\label{max-ball-approximate} {\rm tr}[\varrho^2] \leqslant
\frac{9}{mn+8}.
\end{equation}
\end{proposition}
\begin{proof}
Let ${\rm tr}[\varrho^2] = \mu$. It is not hard to see that in
general $(\lambda_{mn-2} + \lambda_{mn-1} + \lambda_{mn})^2
\leqslant 3(\lambda_{mn-2}^2 + \lambda_{mn-1}^2 + \lambda_{mn}^2)$
and $\lambda_{mn-2}^2 + \lambda_{mn-1}^2 + \lambda_{mn}^2
\leqslant \frac{3}{mn-1} \sum_{i=2}^{mn} \lambda_{i}^2 = 3
\frac{\mu - \lambda_1^2}{mn-1}$. Consequently, if
\begin{equation}
\label{necessary-not-fulfilled} \lambda_1^2 > 9 \, \frac{\mu -
\lambda_1^2}{mn-1},
\end{equation}

\noindent then $\lambda_1 > \lambda_{mn-2} + \lambda_{mn-1} +
\lambda_{mn}$ and the necessary condition for absolute
separability \eqref{necessary} is not fulfilled.

Suppose the purity $\mu$ is known, then the maximal eigenvalue
$\lambda_1$ has a lower bound, which can be found by the method of
Lagrange multipliers with constraints $\sum_{i=1}^{mn} \lambda_i =
1$ and $\lambda_1 \geqslant \lambda_2 \geqslant \ldots \geqslant
\lambda_{mn} \geqslant 0$. The eigenvalue $\lambda_1$ is minimal
if $\lambda_1 = \ldots = \lambda_{k-1}$ and $\lambda_{k+1} =
\lambda_{k+2} = \ldots = 0$ for some $1 < k \leqslant mn$. Then
$\lambda_k = 1 - (k-1)\lambda_1$ and $\mu = (k-1)\lambda_1^2 + [1
- (k-1)\lambda_1]^2$. If $\frac{1}{k} \leqslant \mu \leqslant
\frac{1}{k-1}$, then the minimal largest eigenvalue reads
\begin{equation}
\min \lambda_1 = \frac{1}{k}\left( 1 + \sqrt{\frac{k \mu -
1}{k-1}} \right).
\end{equation}

\noindent Substituting $\min \lambda_1$ for $\lambda_1$ in
equation~\eqref{necessary-not-fulfilled}, we obtain a converse to
inequality \eqref{max-ball-exact}. This converse relation
specifies the region of purities $\mu \in (\mu_0,1]$, where the
state $\varrho \notin \mathcal{A}(m|n)$. Thus,
equation~\eqref{max-ball-exact} is necessary for absolute
separability. Formula \eqref{max-ball-approximate} follows from
equation~\eqref{max-ball-exact} and describes a hyperbola, which
passes through all breaking points of $\mu_0$ as a function of
$mn$, see figure~\ref{figure1}.
\end{proof}

\begin{figure}
\centering
\includegraphics[width=8cm]{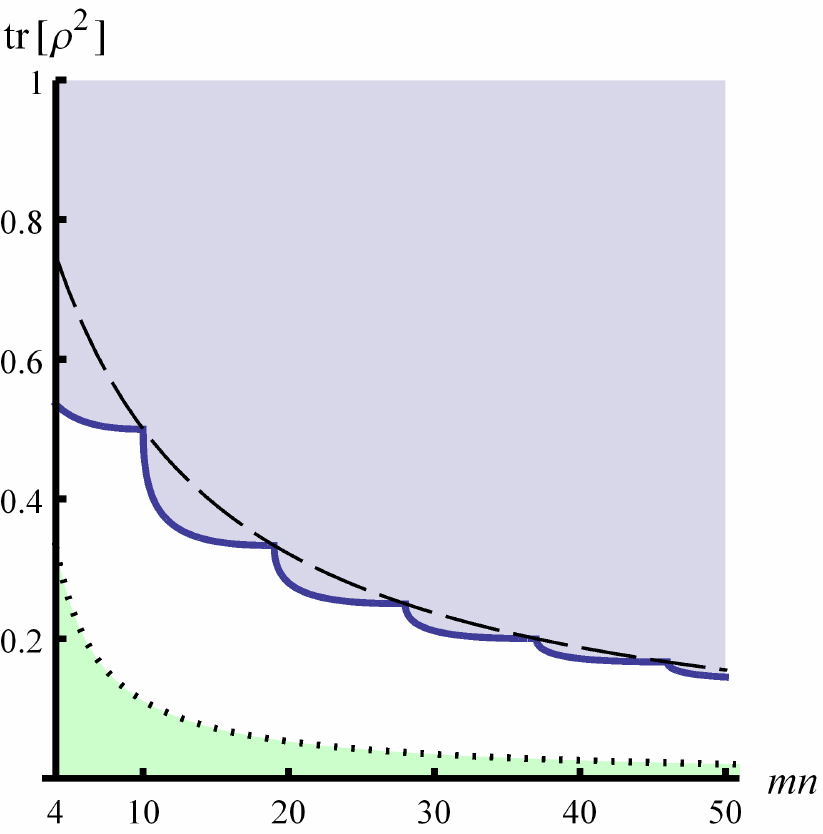}
\caption{\label{figure1} Purity ${\rm tr}[\varrho^2]$ of
states $\varrho \in \mathcal{S}(\mathcal{H}_{mn})$ vs. dimension
$mn$: below dotted line $\varrho \in \mathcal{A}(m|n)$ due to
equation~\eqref{ball}; above solid line $\varrho \notin
\mathcal{A}(m|n)$ due to equation~\eqref{max-ball-exact}. Dashed
line corresponds to the boundary established in
equation~\eqref{max-ball-approximate}.}
\end{figure}

Proposition~\ref{proposition-not-AS-purity} shows, in particular,
that two qubit states with ${\rm tr}[\varrho^2] > (\sqrt{3}-1)^2
\approx 0.536$ cannot be absolutely separable states with respect
to partition $2|2$. A state $\varrho\in\mathcal{S}(\mathcal{H}_d)$
is not absolutely separable with respect to any partition $m|n$
($d=mn \geqslant 4$, $m,n \geqslant 2$) if ${\rm tr}[\varrho^2] >
\frac{9}{d+8}$.

\subsection{Absolute separability with respect to multipartition}

An $N$-partite quantum state $\varrho \in
\mathcal{S}(\mathcal{H}_{n_1 \ldots n_N})$, $n_k \geqslant 2$ is
called fully separable with respect to the partition
$\mathcal{H}_{n_1 \ldots n_N} = \mathcal{H}_{n_1}^{A_1} \otimes
\ldots \otimes \mathcal{H}_{n_N}^{A_N}$ on subsystems
$A_1,\ldots,A_N$ if $\varrho$ adopts the convex sum resolution
$\varrho = \sum_k p_k \varrho_k^{A_1} \otimes \ldots \otimes
\varrho_k^{A_N}$, $p_k \geqslant 0$, $\sum_k p_k = 1$. The
set of fully separable states is denoted by
$\mathcal{S}(\mathcal{H}_{n_1}^{A_1} | \ldots |
\mathcal{H}_{n_N}^{A_N})$. The criterion of full separability is
known, for instance, for 3-qubit Greenberger-Horne-Zeilinger (GHZ)
diagonal states~\cite{chen-2015,han-2016}.

We will call a state $\varrho \in \mathcal{S}(\mathcal{H}_{n_1
\ldots n_N})$ \emph{absolutely separable} with respect to
\emph{multipartition} $n_1| \ldots | n_N$ if $\varrho$ remains
separable with respect to \emph{any} multipartition
$\mathcal{H}_{n_1| \ldots | n_N} = \mathcal{H}_{n_1}^{A_1} \otimes
\ldots \otimes \mathcal{H}_{n_N}^{A_N}$ or, equivalently, $U
\varrho U^{\dag} \in \mathcal{S}(\mathcal{H}_{n_1}^{A_1} | \ldots
| \mathcal{H}_{n_N}^{A_N})$ for any unitary operator $U$ and fixed
multipartition $A_1 | \ldots | A_N$. We will use notation
$\mathcal{A}(n_1 | \ldots |n_N)$ for the set of states, which are
absolutely separable with respect to \emph{multipartition} $n_1|
\ldots | n_N$.

A sufficient condition of absolute separability with respect to
multipartition follows from consideration of separability
balls~\cite{Hildebrand-2007}. Consider an $N$-qubit state $\varrho
\in \mathcal{S}(\mathcal{H}_{2^N})$ such that
\begin{equation}
\label{ball-N-qubits} {\rm tr}[\varrho^2] \leqslant \frac{1}{2^N}
\left( 1 + \frac{54}{17} \, 3^{-N} \right) \, ,
\end{equation}

\noindent then $\varrho \in \mathcal{A}(\underbrace{2 | \ldots
|2}_{N~\text{times}})$~\cite{Hildebrand-2007}.

To illustrate the relation between different types of separability
under bipartitions and multipartitions let us consider a
three-qubit case.

\begin{example}
The inclusion $\mathcal{A}(2|2|2) \subset
\mathcal{S}(\mathcal{H}_{2} | \mathcal{H}_{2} | \mathcal{H}_{2})
\subset \mathcal{S}(\mathcal{H}_{2} | \mathcal{H}_{4}) \subset
\text{PPT}(\mathcal{H}_{2} | \mathcal{H}_{4}) \subset
\mathcal{S}(\mathcal{H}_{8})$ is trivial. Also,
$\mathcal{A}(2|2|2) \subset \mathcal{A}(2|4) = \mathcal{A}_{\rm
PPT}(2|4) \subset \mathcal{S}(\mathcal{H}_{2} | \mathcal{H}_{4})$.
The relation to be clarified is that between $\mathcal{A}(2|4)$
and $\mathcal{S}(\mathcal{H}_{2} | \mathcal{H}_{2} |
\mathcal{H}_{2})$.

Firstly, we notice that the pure state $\ket{\psi_1}\bra{\psi_1}
\otimes \ket{\psi_2}\bra{\psi_2} \otimes \ket{\psi_3}\bra{\psi_3}$
is fully separable but not absolutely separable with respect to
partition $2|4$ as there exists a unitary transformation $U$,
which transforms it into a maximally entangled state $\ket{\rm
GHZ}\bra{\rm GHZ}$, where $\ket{\rm GHZ} = \frac{1}{\sqrt{2}}
(\ket{000}+\ket{111})$. Thus, $\mathcal{S}(\mathcal{H}_{2} |
\mathcal{H}_{2} | \mathcal{H}_{2}) \not\subset \mathcal{A}(2|4)$.

Secondly, consider a state $\varrho \in \mathcal{A}(2|4) =
\mathcal{A}_{\rm PPT}(2|4)$, then its spectrum
$\lambda_1,\ldots,\lambda_8$ in decreasing order satisfies
equation~\eqref{2-n} for $n=4$. Maximizing the state purity
$\sum_{k=1}^8 \lambda_k^2$ under conditions $\lambda_1 \geqslant
\ldots \geqslant \lambda_8 \geqslant 0$, $\sum_{k=1}^8 \lambda_k =
1$, and inequality~\eqref{2-n}, we get $\lambda_1 = \lambda_2 =
\frac{11}{48}$, $\lambda_3 = \frac{23}{144}$, $\lambda_4 =
\lambda_5 = \lambda_6 = \lambda_7 = \lambda_8 = \frac{11}{144}$.
Any 3-qubit state $\varrho$ with such a spectrum is absolutely
separable with respect to partition $2|4$. Consider a particular
state
\begin{equation}
\label{rhoASnotS} \varrho=\sum_{k=1}^8 \lambda_k |{\rm GHZ}_k
\rangle \langle {\rm GHZ}_k|,
\end{equation}

\noindent where the binary representation of $k - 1 = 4 k_1 + 2
k_2 + k_3$, $k_i = 0,1$, defines GHZ-like states
\begin{equation}
| {\rm GHZ}_k \rangle = \frac{1}{\sqrt{2}} ((-1)^{k-1}
|k_1\rangle\otimes|k_2\rangle\otimes|k_3\rangle +
|\bar{k_1}\rangle\otimes|\bar{k_2}\rangle\otimes|\bar{k_3}\rangle
\end{equation}

\noindent with $\bar{0} = 1$ and $\bar{1} = 0$. The state
\eqref{rhoASnotS} is GHZ diagonal, so we apply to it the necessary
and sufficient condition of full separability~(\cite{han-2016},
theorem 5.2), which shows that \eqref{rhoASnotS} is not fully
separable. Thus, $\mathcal{A}(2|4) \not\subset
\mathcal{S}(\mathcal{H}_{2} | \mathcal{H}_{2} | \mathcal{H}_{2})$.

Finally, to summarize the results of this example, we depict the
Venn diagram of separable and absolutely separable 3 qubit states
in figure~\ref{figure2}.
\end{example}

\begin{figure}
\centering
\includegraphics[width=8cm]{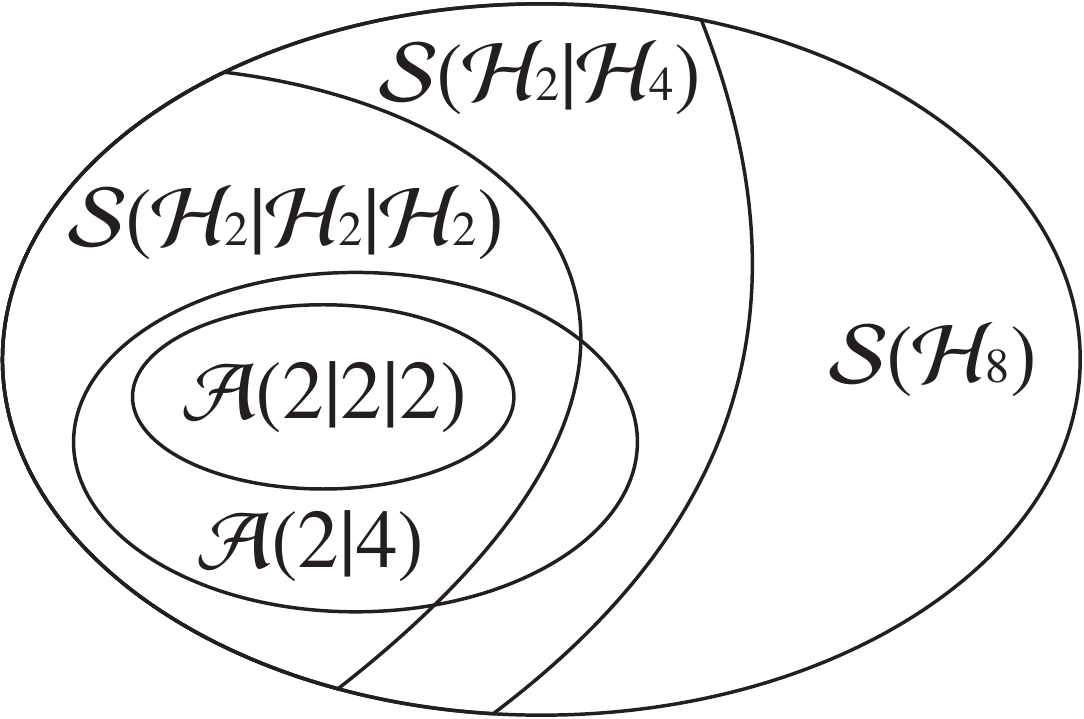}
\caption{\label{figure2} The relation between separability classes
of three qubit states: $\mathcal{S}(\mathcal{H}_8)$ is the set of
density operators, $\mathcal{S}(\mathcal{H}_2|\mathcal{H}_4)$ is
the set of states separable with respect to a fixed bipartition
$\mathcal{H}_2|\mathcal{H}_4$,
$\mathcal{S}(\mathcal{H}_2|\mathcal{H}_2|\mathcal{H}_2)$ is the
set of fully separable states with respect to a multipartition
$\mathcal{H}_2|\mathcal{H}_2|\mathcal{H}_2$, $\mathcal{A}(2|4)$ is
a set of absolutely separable states with respect to partition
$2|4$, and $\mathcal{A}(2|2|2)$ is a set of absolutely separable
states with respect to partition $2|2|2$. Convex figures
correspond to convex sets.}
\end{figure}

Note that the state $\varrho$ in \eqref{rhoASnotS} is separable
for \emph{any} bipartition $2|4$ and entangled with respect to
multipartition $\mathcal{H}_{2} | \mathcal{H}_{2} |
\mathcal{H}_{2}$. In particular, $\varrho$ is separable with
respect to bipartitions $\mathcal{H}_{2}^{A} |
\mathcal{H}_{4}^{BC}$, $\mathcal{H}_{2}^{B} |
\mathcal{H}_{4}^{AC}$, and $\mathcal{H}_{2}^{C} |
\mathcal{H}_{4}^{AB}$, but entangled with respect to tripartition
$\mathcal{H}_{2}^{A} | \mathcal{H}_{2}^{B} | \mathcal{H}_{2}^{C}$.
The states with such a property were previously constructed via
unextendable product bases~\cite{bennett-1999,divincenzo-2003}.
Note, however, that even if a 3 qubit state $\xi$ is separable
with respect to the specific partitions $A|BC$, $B|AC$, and
$C|AB$, it does not imply that $\xi$ is absolutely separable with
respect to partition $2|4$, because $U \xi U^\dag$ is separable
with respect to $A|BC$ only for permutation matrices $U(A
\leftrightarrow B)$, $U(B \leftrightarrow C)$, and
$U(A\leftrightarrow C)$, but not general unitary operators $U$.

\section{Absolutely separating maps and channels}
\label{section-as-channels}

In quantum information theory, positive linear maps
$\Phi:\mathcal{S}(\mathcal{H}) \mapsto \mathcal{S}(\mathcal{H})$
represent a useful mathematical tool in characterization of
bipartite entanglement~\cite{horodecki-1996}, multipartite
entanglement~\cite{filippov-melnikov-ziman-2013,F-M},
characterization of Markovianity in open system
dynamics~\cite{rivas-2010,benatti-2017}, etc. A quantum channel is
given by a CPT map $\Phi$ such that $\Phi \otimes {\rm Id}_k$ is a
positive map for all identity transformations ${\rm Id}_k:
\mathcal{S}(\mathcal{H}_k) \mapsto \mathcal{S}(\mathcal{H}_k)$.
Thus, entanglement-related properties are easier to explore for
positive maps~\cite{filippov-ziman-2013} but deterministic
physical evolutions are given by quantum channels. It means that
the set of absolutely separating channels is the intersection of
CPT maps with the set of positive absolutely separating maps
introduced below.

We recall that a linear map $\Phi:
\mathcal{S}(\mathcal{H}_{mn}) \mapsto
\mathcal{S}(\mathcal{H}_{m}|\mathcal{H}_{n})$ is called positive
entanglement annihilating with respect to partition
$\mathcal{H}_{m}|\mathcal{H}_{n}$, concisely,
PEA($\mathcal{H}_{m}|\mathcal{H}_{n}$). For multipartite composite
systems, $\Phi: \mathcal{S}(\mathcal{H}_{n_1 \ldots n_N}) \mapsto
\mathcal{S}(\mathcal{H}_{n_1}| \ldots | \mathcal{H}_{n_N})$ is
positive entanglement annihilating, PEA($\mathcal{H}_{n_1}| \ldots
|\mathcal{H}_{n_N}$). The map $\Phi: \mathcal{S}(\mathcal{H}_{m})
\mapsto \mathcal{S}(\mathcal{H}_{m})$ is called entanglement
breaking (EB) if $\Phi \otimes {\rm Id}_n$ is positive
entanglement annihilating for all
$n$~\cite{holevo-1998,king-2002,shor-2002,ruskai-2003,horodecki-2003}.
Note that an EB map is automatically completely positive, which
means that any EB map is a quantum channel (CPT map).

In this paper, we focus on positive absolutely separating maps
$\Phi: \mathcal{S}(\mathcal{H}_{mn}) \mapsto \mathcal{A}(m|n)$,
whose output is always absolutely separable for valid input
quantum states. We will denote such maps by PAS($m|n$). Clearly,
$\text{PAS}(m|n) \subset
\text{PEA}(\mathcal{H}_{m}|\mathcal{H}_{n})$. Absolutely
separating channels with respect to partition $m|n$ are the maps
$\Phi \in \text{CPT} \cap \text{PAS}(m|n)$. Note that the concept
of absolutely separating map can be applied not only to linear
positive maps but also to non-linear physical maps originating in
measurement procedures, see e.g.~\cite{luchnikov-filippov-2017}.
In this paper, however, we restrict to linear maps only.

Let us notice that the application of any positive map $\Phi:
\mathcal{S}(\mathcal{H}_{n}) \mapsto \mathcal{S}(\mathcal{H}_{n})$
to a part of composite system cannot result in an absolutely
separating map.

\begin{proposition}
\label{proposition-Phi-Id} The map $\Phi \otimes \text{Id}_n$ is
not absolutely separating with respect to partition $m|n$ for any
positive map $\Phi: \mathcal{S}(\mathcal{H}_m) \mapsto
\mathcal{S}(\mathcal{H}_m)$, $n \geqslant 2$.
\end{proposition}
\begin{proof}
Consider the input state $\varrho_{\rm in} = \varrho_1 \otimes
|\psi\rangle\langle\psi|$, then the output state is $\varrho_{\rm
out} = \Phi[\varrho_1] \otimes \ket{\psi}\bra{\psi}$. Spectrum of
$\varrho_{\rm out}$ does not satisfy the necessary condition of
absolute separability, equation \eqref{necessary}, so $\Phi
\otimes \text{Id}_n$ is not absolutely separating.
\end{proof}

The physical meaning of proposition~\ref{proposition-Phi-Id} is
that there exists no local action on a part of quantum system,
which would make all outcome quantum states absolutely separable.
This is in contrast with separability property since entanglement
breaking channels disentangle the part they act on from other
subsystems. Proposition~\ref{proposition-Phi-Id} means that
one-sided quantum noises $\Phi \otimes {\rm Id}$ can always be
compensated by a proper choice of input state $\varrho$ and
unitary operations $U$ in such a way that the outcome state $U (
\Phi\otimes{\rm Id}[\varrho] ) U^{\dag}$ becomes entangled.

It was emphasized already that the absolutely separable state can
be transformed into an entangled one only by non-unitary maps.
However, not every non-unitary map is adequate for entanglement
restoration. For instance, unital quantum channels cannot result
in entangled output for absolutely separable input.

\begin{proposition}
\label{proposition-concatenation} Suppose $\Phi_1$ is absolutely
separating channel with respect to some (multi)partition and $\Phi_2$ is a
unital channel, i.e. $\Phi_2[I] = I$. Then the concatenation $\Phi_2 \circ \Phi_1$ is
also absolutely separating with respect to the same partition.
\end{proposition}
\begin{proof}
From absolute separability of $\Phi_1$ it follows that $\varrho =
\Phi_1[\varrho_{\text{in}}]$ is absolutely separable for any input
$\varrho_{\rm in}$. Since the channel $\Phi_2$ is unital,
$\Phi_2[\varrho] \prec \varrho$ for any density operator
$\varrho$~\cite{Uhlmann}, i.e. the ordered spectrum of
$\Phi_2[\varrho]$ is majorized by the ordered spectrum of
$\varrho$, with $\varrho$ being absolutely separable in our case.
Thus, the spectrum of the state
$\Phi_2\circ\Phi_1[\varrho_{\text{in}}]$ is majorized by the
spectrum of the absolutely separable state and according to Lemma
2.2 in $\cite{Jivulescu}$ this implies absolute separability of
$\Phi_2\circ\Phi_1[\varrho_{\text{in}}]$.
\end{proof}

There exist such physical maps $\Phi:\mathcal{S}(\mathcal{H}_d)
\mapsto \mathcal{S}(\mathcal{H}_d)$ that are not sensitive to
unitary rotations of input states and translate that property to
the output states. We will call the map
$\Phi:\mathcal{S}(\mathcal{H}_d) \mapsto
\mathcal{S}(\mathcal{H}_d)$ covariant if
\begin{equation}
\label{covariant} \Phi[U \varrho U^{\dag}] = U \Phi[\varrho]
U^{\dag}
\end{equation}

\noindent for all $U \in \text{SU(d)}$. The example of covariant
map is the depolarizing channel
$\mathcal{D}_q:\mathcal{S}(\mathcal{H}_d) \mapsto
\mathcal{S}(\mathcal{H}_d)$ acting as follows:
\begin{equation}
\label{depolarizing} \mathcal{D}_q [X] = q X + (1-q) {\rm tr}[X]
\frac{1}{d} I_d,
\end{equation}

\noindent which is completely positive if $q \in [-1/(d^2-1),1]$.

\begin{proposition}
A covariant map $\Phi: \mathcal{S}(\mathcal{H}_{mn}) \mapsto
\mathcal{S}(\mathcal{H}_{mn})$ is absolutely separating with
respect to partition $m|n$ if and only if it is entanglement
annihilating with respect to partition
$\mathcal{H}_m|\mathcal{H}_n$.
\end{proposition}
\begin{proof}
Suppose $\Phi$ is covariant and entanglement annihilating. Since
$\Phi$ is entanglement annihilating, then the left hand side of
equation~\eqref{covariant} is separable for all $U$ with respect to
partition $\mathcal{H}_m|\mathcal{H}_n$. Due to covariance
property it means that $U \Phi[\varrho] U^{\dag} \in
\mathcal{S}(\mathcal{H}_m|\mathcal{H}_n)$ for all unitary $U$,
i.e. $\Phi$ is PAS($m|n$).

Suppose $\Phi$ is covariant and absolutely separating with respect
to partition $m|n$. Consider pure states $\varrho =
\ket{\psi}\bra{\psi} \in \mathcal{S}(\mathcal{H}_{mn})$. Since
$\Phi$ is absolutely separating, the right hand side of
equation~\eqref{covariant} is separable with respect to a fixed
partition $\mathcal{H}_m|\mathcal{H}_n$ for all $U$. By covariance
this implies $\Phi[U \ket{\psi}\bra{\psi}U^{\dag}] \in
\mathcal{S}(\mathcal{H}_m|\mathcal{H}_n)$ for all unitary $U$,
i.e. $\Phi[\ket{\varphi}\bra{\varphi}] \in
\mathcal{S}(\mathcal{H}_m|\mathcal{H}_n)$ for all pure states
$\ket{\varphi}$. Since the set of input states
$\mathcal{S}(\mathcal{H}_{mn})$ is convex, it implies that
$\Phi[\varrho_{\rm in}] \in
\mathcal{S}(\mathcal{H}_m|\mathcal{H}_n)$ for all input states
$\varrho_{\rm in}$, i.e. $\Phi$ is entanglement annihilating with
respect to partition $\mathcal{H}_m|\mathcal{H}_n$.
\end{proof}

\begin{example}
\label{example-global-depolarizing-as} The depolarizing channel
$\mathcal{D}_q:\mathcal{S}(\mathcal{H}_{mn}) \mapsto
\mathcal{S}(\mathcal{H}_{mn})$ is known to be
$\text{PEA}(\mathcal{H}_m | \mathcal{H}_n)$ if $q \leqslant
\frac{2}{mn+2}$~\cite{filippov-2014,lami-huber-2016}. Therefore,
$\mathcal{D}_q$ is absolutely separating with respect to partition
$m|n$ if $q \leqslant \frac{2}{mn+2}$ because $\mathcal{D}_q$ is
covariant.
\end{example}

The following results show the behaviour of absolutely separating
maps under tensor product.

\begin{proposition}
Suppose $\Phi_1: \mathcal{S}(\mathcal{H}_{m_1 n_1}) \mapsto
\mathcal{S}(\mathcal{H}_{m_1 n_1})$ and $\Phi_2:
\mathcal{S}(\mathcal{H}_{m_2 n_2}) \mapsto
\mathcal{S}(\mathcal{H}_{m_2 n_2})$ are such positive maps that
$\Phi = \Phi_1 \otimes \Phi_2$ is absolutely separating with
respect to partition $m_1 m_2 | n_1 n_2$. Then $\Phi_1$ is
$\text{PAS}(m_1|n_1)$ and $\Phi_2$ is $\text{PAS}(m_2|n_2)$.
\end{proposition}
\begin{proof}
Let $\varrho_{\rm in} = \varrho_1 \otimes \varrho_2$, where
$\varrho_1 \in \mathcal{S}(\mathcal{H}_{m_1 n_1})$ and $\varrho_2
\in \mathcal{S}(\mathcal{H}_{m_2 n_2})$, then
\begin{equation}
\label{U-rotated-factorized} U \Phi (\varrho_1 \otimes \varrho_2)
U^{\dagger} = U \Phi_1(\varrho_1) \otimes \Phi_2(\varrho_2)
U^{\dagger}
\end{equation}

\noindent is separable with respect to a specific bipartition
$\mathcal{H}_{m_1 m_2}^{AB}|\mathcal{H}_{n_1 n_2}^{CD}$ for any
unitary operator $U$. So the state \eqref{U-rotated-factorized}
can be written as
\begin{equation}
U \Phi_1 (\varrho_1) \otimes \Phi_2 (\varrho_2) U^{\dagger} =
\sum_k p_k \varrho_{k}^{AB} \otimes \varrho_{k}^{CD}.
\end{equation}

\noindent Tracing out subsystem $BD$ we get
\begin{equation}
{\rm tr}_{BD} \left( \sum_k p_k \varrho_{k}^{AB} \otimes
\varrho_{k}^{CD} \right) = \sum_k p_k \varrho_k^A \otimes
\varrho_k^C,
\end{equation}

\noindent which is separable with respect to bipartition $A|C$.
Suppose $U = U_1 \otimes U_2$ in \eqref{U-rotated-factorized},
then we obtain that $U_1 \Phi_1(\varrho_1) U_1^{\dagger}$ is
separable with respect to bipartition $A|C$ for all $U_1$, which
means that $\Phi_1$ is $\text{PAS}(m_1|n_1)$. By the same line of
reasoning, $\Phi_2$ is $\text{PAS}(m_2|n_2)$.
\end{proof}

However, even if two maps $\Phi_1 \in \text{PAS}(m_1|n_1)$ and
$\Phi_2 \in \text{PAS}(m_2|n_2)$, the map $\Phi_1 \otimes \Phi_2$
can still be not absolutely separable with respect to partition
$m_1 m_2 | n_1 n_2$, which is illustrated by the following
example.

\begin{example}\label{ASdepolarizingglobal}
Consider a four qubit map $\Phi: \mathcal{S}(\mathcal{H}_{16})
\mapsto \mathcal{S}(\mathcal{H}_{16})$ of the form $\Phi = {\cal
D}_q \otimes {\cal D}_q$, where ${\cal D}_q:
\mathcal{S}(\mathcal{H}_{4}) \mapsto \mathcal{S}(\mathcal{H}_{4})$
is a two qubit global depolarizing channel given by
equation~\eqref{depolarizing}. Let $q = \frac{1}{3}$ then ${\cal
D}_{1/3}$ is absolutely separating with respect to partition $2|2$
by example~\ref{example-global-depolarizing-as}. Despite the fact
that both parts of the tensor product ${\cal D}_{1/3} \otimes
{\cal D}_{1/3}$ are absolutely separating with respect to $2|2$,
$\Phi$ is not absolutely separating with respect to $4|4$. In
fact, let $U
= \begin{pmatrix} I_7 & 0 & 0 & 0 \\
0 & \frac{1}{\sqrt{2}} & -\frac{\rmi}{\sqrt{2}} & 0 \\
0 & \frac{\rmi}{\sqrt{2}} & \frac{1}{\sqrt{2}} & 0 \\
0 & 0 & 0 & I_7
\end{pmatrix}$ be a $16\times 16$ unitary matrix in the conventional four-qubit basis, $\varrho = (\ket{\psi}\bra{\psi})^{\otimes 2}$, $|\psi\rangle =
\frac{1}{\sqrt{2}} (|00\rangle+|11\rangle)$, then $U \Phi[\varrho]
U^{\dagger}$ is entangled with respect to partition
$\mathcal{H}_4|\mathcal{H}_4$ because the the partially transposed
output density matrix $(U \Phi[\varrho] U^{\dagger})^{\Gamma}$ has
negative eigenvalue $\lambda < -0.0235$. Thus, $\Phi = {\cal
D}_{1/3}\otimes{\cal D}_{1/3}$ is not absolutely separating with
respect to partition $4|4$ even though each ${\cal D}_{1/3}$ is
absolutely separating with respect to partition $2|2$.
\end{example}

The practical criterion to detect absolutely separating channels
follows from the consideration of norms. Let us recall that for a
given linear map $\Phi$ and real numbers $1\leq p, q\leq\infty$,
the induced Schatten superoperator norm \cite{KR04, KNR05, WAT05}
of $\Phi$ is defined by formula
\begin{equation}
\label{Schatten} \|\Phi\|_{q \rightarrow
p}:=\sup\limits_{X}\bigl\{\|\Phi[X]\|_p:\|X\|_q = 1\bigr\},
\end{equation}

\noindent where $\|\cdot\|_p$ and $\|\cdot\|_q$ are the Schatten
$p$- and $q$-norms, i.e. $\|A\|_p=\left[{\rm
tr}\left((A^{\dagger}A)^{\frac{p}{2}}\right)\right]^{\frac{1}{p}}$.
Physically, in the case $q=1$ and $p=2$ equation~\eqref{Schatten}
provides the maximal output purity $(\|\Phi\|_{1 \rightarrow 2})^2
= \max_{\varrho \in \mathcal{S}(\mathcal{H})} {\rm
tr}[(\Phi[\varrho])^2]$.

\begin{proposition}
\label{proposition-ball} A positive linear map
$\Phi:\mathcal{S}(\mathcal{H}_{mn}) \mapsto
\mathcal{S}(\mathcal{H}_{mn})$ is absolutely separating with
respect to partition $m|n$ if
\begin{equation}
\label{ASmap} (\|\Phi\|_{1\rightarrow2})^2 \leqslant
\frac{1}{mn-1}.
\end{equation}
\end{proposition}
\begin{proof}
If \eqref{ASmap} holds, then the state $\Phi[\varrho]$ satisfies
equation~\eqref{ball} and belongs to the separability ball, i.e.
$\Phi[\varrho]$ is absolutely separable with respect to partition
$m|n$ for all $\varrho \in \mathcal{S}(\mathcal{H}_{mn})$.
\end{proof}

\begin{proposition}
\label{proposition-N-qubits} A positive linear map
$\Phi:\mathcal{S}(\mathcal{H}_{2^N}) \mapsto
\mathcal{S}(\mathcal{H}_{2^N})$ is absolutely separating with
respect to partition $\underbrace{2 | \ldots |2}_{N~\text{times}}$
if
\begin{equation}
\label{ASmap-N-qubits} (\|\Phi\|_{1\rightarrow2})^2 \leqslant
\frac{1}{2^N} \left( 1 + \frac{54}{17} \, 3^{-N} \right).
\end{equation}
\end{proposition}
\begin{proof}
If \eqref{ASmap-N-qubits} holds, then the $N$-qubit state
$\Phi[\varrho]$ satisfies equation~\eqref{ball-N-qubits} and
belongs to the full separability ball, i.e. $\Phi[\varrho]$ is
absolutely separable with respect to partition $2 | \ldots |2$ for
all $\varrho \in \mathcal{S}(\mathcal{H}_{2^N})$.
\end{proof}

A necessary condition for the map
$\Phi:\mathcal{S}(\mathcal{H}_{mn}) \mapsto
\mathcal{S}(\mathcal{H}_{mn})$ to be absolutely separating with
respect to partition $m|n$ follows from equation~\eqref{necessary}
which must be satisfied by all output states $\Phi[\varrho]$. If a
map has a local structure, $\Phi = \Phi_1 \otimes \Phi_2$, then
the output state $\Phi[\varrho_1 \otimes \varrho_2] =
\Phi_1[\varrho_1] \otimes \Phi_2[\varrho_2]$ is factorized for
factorized input states $\varrho_1 \otimes \varrho_2$.

\begin{proposition}
\label{proposition-not-AS} A local map $\Phi_1 \otimes \Phi_2$
with $\Phi_1:\mathcal{S}(\mathcal{H}_{m}) \mapsto
\mathcal{S}(\mathcal{H}_{m})$ and
$\Phi_2:\mathcal{S}(\mathcal{H}_{n}) \mapsto
\mathcal{S}(\mathcal{H}_{n})$ is not absolutely separating with
respect to partition $m|n$ if the image (range) of $\Phi_1$ or
$\Phi_2$ contains a boundary point of
$\mathcal{S}(\mathcal{H}_{m})$ or $\mathcal{S}(\mathcal{H}_{n})$,
respectively.
\end{proposition}
\begin{proof}
Suppose the image of $\Phi_1$ contains a boundary point of
$\mathcal{S}(\mathcal{H}_{m})$, i.e. there exists a state
$\varrho_1$ such that $\Phi_1[\varrho_1] \in \partial
\mathcal{S}(\mathcal{H}_{m})$, then $\Phi_1[\varrho_1] \otimes
\Phi_2[\varrho_2]$ is not absolutely separating with respect to
partition $m|n$, see the discussion after equation~\eqref{necessary}.
Analogous proof takes place if the image of $\Phi_2$ contains a
boundary point of $\mathcal{S}(\mathcal{H}_{n})$.
\end{proof}

\begin{example}
Suppose $\Phi_1:\mathcal{S}(\mathcal{H}_m) \mapsto
\mathcal{S}(\mathcal{H}_m)$ is an amplitude damping
channel~\cite{nielsen-2000} and $\Phi_2:\mathcal{S}(\mathcal{H}_n)
\mapsto \mathcal{S}(\mathcal{H}_n)$ is an arbitrary channel, then
$\Phi_1 \otimes \Phi_2$ is not absolutely separating with respect
to $m|n$ by proposition~\ref{proposition-not-AS}, because $\Phi_1$
has a fixed point, which is a pure state.
\end{example}

Similarly, if the maximal output purity of a positive map is large
enough, then it cannot be absolutely separating.

\begin{proposition}
\label{proposition-anti-ball} A positive linear map
$\Phi:\mathcal{S}(\mathcal{H}_{mn}) \mapsto
\mathcal{S}(\mathcal{H}_{mn})$ is not absolutely separating with
respect to partition $m|n$ if
\begin{equation}
\label{notASmap} (\|\Phi\|_{1\rightarrow 2})^2 > \frac{9}{mn+8}.
\end{equation}
\end{proposition}
\begin{proof}
Inequality \eqref{notASmap} implies that there exists a state
$\varrho\in\mathcal{S}(\mathcal{H}_{mn})$ such that the output
state $\Phi[\varrho]$ violates
inequality~\eqref{max-ball-approximate}, i.e. $\Phi[\varrho]$ is
not absolutely separable with respect to partition $m|n$ and the
map $\Phi$ is not absolutely separating.
\end{proof}

\section{Tensor-stable absolutely separating maps}
\label{section-n-tensor-stable-as}

Suppose a map $\Phi:\mathcal{S}(\mathcal{H}_d) \mapsto
\mathcal{S}(\mathcal{H}_d)$. We will refer to $\Phi$ as
$N$-tensor-stable absolutely separating if $\Phi^{\otimes N}$ is
absolutely separating with respect to any valid partitions $m|n$,
$d^N=mn$, $m,n \geqslant 2$. If $\Phi$ is $N$-tensor-stable
absolutely separating for all $N=1,2,\ldots$, then $\Phi$ is
called tensor-stable absolutely separating. These definitions are
inspired by the paper~\cite{muller-hermes-2016}, where the
stability of positive maps under tensor product was studied.

In what follows we show that all $N$-tensor-stable absolutely
separating maps $\Phi:\mathcal{S}(\mathcal{H}_d) \mapsto
\mathcal{S}(\mathcal{H}_d)$ are close to the tracing map ${\rm
Tr}[\varrho] = {\rm tr}[\varrho] \frac{I_d}{d}$. To quantify such
a closeness, one can use either the maximal output purity
$(\|\Phi\|_{1\rightarrow 2})^2$ or the minimal output
entropy~\cite{holevo-2012}:
\begin{equation}
h(\Phi) = \min_{\varrho \in \mathcal{S}(\mathcal{H}_d)} {\rm
tr}\Big[ - \Phi[\varrho] \log \Phi[\varrho] \Big],
\end{equation}

\noindent where $\log$ stands for the natural logarithm. Note that
$\frac{1}{d} \leqslant (\|\Phi\|_{1\rightarrow 2})^2 \leqslant 1$
and $0 \leqslant h(\Phi) \leqslant \log d$. Since
$(\|\Phi\|_{1\rightarrow 2})^2 = \frac{1}{d}$ and $h(\Phi) = \log
d$ if and only if $\Phi = {\rm Tr}$, the differences
$(\|\Phi\|_{1\rightarrow 2})^2 - \frac{1}{d}$ and $\log d -
h(\Phi)$ can be interpreted as the measure of closeness between
maps $\Phi$ and ${\rm Tr}$.

\begin{proposition}
A map $\Phi:\mathcal{S}(\mathcal{H}_d) \mapsto
\mathcal{S}(\mathcal{H}_d)$ is not $N$-tensor-stable absolutely
separating if
\begin{equation}
\label{N-larger-purity} N > \frac{8}{d(\|\Phi\|_{1\rightarrow
2})^2 - 1} + 1
\end{equation}

\noindent or
\begin{equation}
\label{N-larger-entropy} N > 8 \left( \frac{\log d + 1}{\log d -
h(\Phi)} \right)^2 + 1.
\end{equation}
\end{proposition}
\begin{proof}
Suppose the map $\Phi^{\otimes N}$ and the input state
$\varrho^{\otimes N}$, then $\Phi^{\otimes N}[\varrho^{\otimes N}]
= \left( \Phi[\varrho] \right)^{\otimes N}$. Let decreasingly
ordered eigenvalues of $\Phi[\varrho]$ be
$\lambda_1,\ldots,\lambda_d$, then the decreasingly ordered
eigenvalues $\Lambda_1,\ldots,\Lambda_{d^N}$ of $\left(
\Phi[\varrho] \right)^{\otimes N}$ satisfy the following
relations:
\begin{equation}
\lambda_1^N = \Lambda_1, \qquad \lambda_1 \lambda_d^{N-1}
\geqslant \Lambda_{d^N - 2}, \qquad \lambda_1 \lambda_d^{N-1}
\geqslant \Lambda_{d^N - 1}, \qquad \lambda_1 \lambda_d^{N-1}
\geqslant \lambda_d^{N} = \Lambda_{d^N}.
\end{equation}

If $\Lambda_1 > \Lambda_{d^N - 2} + \Lambda_{d^N - 1} +
\Lambda_{d^N}$, then $\left( \Phi[\varrho] \right)^{\otimes N}$ is
not absolutely separable with respect to any partition in view of
equation~\eqref{necessary} and $\Phi$ is not $N$-tensor-stable
absolutely separating. On the other hand, inequality $\Lambda_1 >
\Lambda_{d^N - 2} + \Lambda_{d^N - 1} + \Lambda_{d^N}$ follows
from the inequalities $\lambda_1^{N-1} > 3 \lambda_d^{N-1}$ and
$(d\lambda_1)^{N-1} > 3$ because $\frac{1}{d} \geqslant
\lambda_d$.

Let $\varrho$ be a state, which maximizes the purity of
$\Phi[\varrho]$, then $(\|\Phi\|_{1\rightarrow 2})^2 =
\sum_{i=1}^d \lambda_i^2$ and $\lambda_1^2 \geqslant \frac{1}{d}
(\|\Phi\|_{1\rightarrow 2})^2$. Consequently, $(d\lambda_1)^2
\geqslant d (\|\Phi\|_{1\rightarrow 2})^2$ and the inequality
\begin{equation}
\label{aux-purity-N-bound} d (\|\Phi\|_{1\rightarrow 2})^2 > 1 +
\frac{8}{N-1} > \sqrt[N-1]{9}
\end{equation}

\noindent implies $(d\lambda_1)^{N-1} > 3$. Finally, the first
inequality in equation~\eqref{aux-purity-N-bound} is equivalent to
inequality~\eqref{N-larger-purity} and provides a sufficient
condition for the map $\Phi$ not to be $N$-tensor-stable
absolutely separable.

Let $\varrho$ be a state, which minimizes the entropy of
$\Phi[\varrho]$, then $h(\Phi) = - \sum_{i=1}^d \lambda_i \log
\lambda_i$. Denote  $T=\|\Phi[\varrho] - \frac{1}{d} I\|_1 =
\sum_{i=1}^d |\lambda_i - \frac{1}{d}| < 2$, then $\frac{T}{2} =
\sum_{i:\ \lambda_i \geqslant \frac{1}{d}} (\lambda_i -
\frac{1}{d}) \leqslant d (\lambda_1 - \frac{1}{d})$ and $\lambda_1
\geqslant \frac{1}{d}\left( 1 + \frac{T}{2} \right)$. Using
results of the paper~\cite{fannes-1973}, we obtain
\begin{equation}
\log d - h(\Phi) \leqslant T \log d + \min\left(-T \log
T,\frac{1}\rme\right) \leqslant T \log d + \sqrt{2T} \leqslant
\sqrt{2T} (\log d + 1).
\end{equation}

\noindent Therefore,
\begin{equation}
\label{aux-entropy-N-bound} (d\lambda_1)^{N-1} \geqslant \left( 1
+ \frac{T}{2} \right)^{N-1} \geqslant 1 + \frac{N-1}{2} T
\geqslant 1 + \frac{N-1}{4} \left( \frac{\log d - h(\Phi)}{\log d
+ 1} \right)^2.
\end{equation}

\noindent If inequality \eqref{N-larger-entropy} is fulfilled,
then the right hand side of equation~\eqref{aux-entropy-N-bound}
is greater than 3, which implies $(d\lambda_1)^{N-1} > 3$ and
$\Phi$ is not $N$-tensor-stable absolutely separating.
\end{proof}

If $\Phi \neq {\rm Tr}$, then there exists $N$ such that
$\Phi^{\otimes N}$ is not absolutely separating. On the contrary,
if $\Phi = {\rm Tr}$, then $\Phi^{\otimes N}$ is absolutely
separating with respect to any partition for all $N$ because
$\Phi^{\otimes N}[\widetilde{\varrho}] = \frac{1}{d^N} I_{d^N}$
for all $\widetilde{\varrho}$.

\begin{corollary}
A map $\Phi:\mathcal{S}(\mathcal{H}_d) \mapsto
\mathcal{S}(\mathcal{H}_d)$ is tensor-stable absolutely separating
if and only if $\Phi={\rm Tr}$.
\end{corollary}

Physical interpretation of this result can be also based on the
fact that $\varrho^{\otimes N}$ allows Schumacher
compression~\cite{schumacher-1995}, namely, $\varrho^{\otimes N}
\approx P \oplus 0 \approx \ket{\psi}\bra{\psi} \otimes P$, where
$P$ is a projector onto the typical subspace of dimension
$\rme^{S(\varrho)N}$ and $\ket{\psi} \in \mathcal{H}_{\rme^{[\log
d - S(\varrho)]N}}$. If $S(\varrho) \neq \log d$, then for
sufficiently large number $N$ of identical mixed states $\varrho$
the dimension $\rme^{[\log d - S(\varrho)]N}$ exceeds 4, so
$\ket{\psi}$ can be transformed into an entangled state
$U\ket{\psi}$ by the action of a proper unitary operator $U$.

\section{Specific absolutely separating maps and channels}
\label{section-specific-as-channels}

In this section we focus on particular physical evolutions and
transformations, which either describe specific dynamical maps or
represent interesting examples of linear state transformations. We
characterize the region of parameters, where the map is absolutely
separating and find states robust to the loss of property to be
not absolutely separable.

\subsection{Local depolarizing qubit maps and channels}
\label{subsection-depolarizing}

Let us analyze a map of the form
$\mathcal{D}_{q_1}\otimes\mathcal{D}_{q_2}$, where
\begin{equation}
\mathcal{D}_{q}[X] = q X + (1-q) \text{tr}[X] \frac{1}{2}I.
\end{equation}
Map $\mathcal{D}_q$ is positive for $q\in[-1,1]$ and completely
positive if $q \in [-\frac{1}{3},1]$. As absolutely separating
maps are the subset of entanglement annihilating maps, it is worth
to mention that entanglement-annihilating properties of the map
$\mathcal{D}_{q_1}\otimes\mathcal{D}_{q_2}$ and their
generalizations (acting in higher dimensions) are studied in the
papers~\cite{filippov-2014,filippov-ziman-2013,lami-huber-2016}.

Since depolarizing maps are not sensitive to local changes of
basis states, we consider a pure input state
$|\psi\rangle\langle\psi|$, where $|\psi\rangle$ always adopts the
Schmidt decomposition
$|\psi\rangle=\sqrt{p}|00\rangle+\sqrt{1-p}|11\rangle$ in the
proper local bases. We denote $\varrho_{\rm out} =
\mathcal{D}_{q_1} \otimes \mathcal{D}_{q_2}
[|\psi\rangle\langle\psi|]$. Using
proposition~\ref{proposition-ball}, we conclude that
$\mathcal{D}_{q_1}\otimes\mathcal{D}_{q_2}$ is absolutely
separating with respect to partition $2|2$ if
$\text{tr}[\varrho_{\rm out}^2] \leqslant \frac{1}{3}$ for all $p
\in [0,1]$, which reduces to
\begin{equation}
\label{AbsSepLocalDep}
q_1^2 + q_2^2 + q_1^2 q_2^2 \leqslant \frac{1}{3}.
\end{equation}

Note that equation~\eqref{AbsSepLocalDep} provides only sufficient
condition for absolutely separating maps $\mathcal{D}_{q_1}
\otimes \mathcal{D}_{q_2}$. The area of parameters $q_1,q_2$
satisfying equation~\eqref{AbsSepLocalDep} is depicted in
figure~\ref{figure3}.

\begin{proposition}
Two-qubit local depolarizing map $\mathcal{D}_{q_1} \otimes
\mathcal{D}_{q_2}$ is absolutely separating with respect to
partition $2|2$ if and only if
\begin{eqnarray}\
&& \label{local-dep-1} q_1(1+|q_2|) \leqslant \sqrt{1-q_1^2}\, (1-|q_2|) ~~\text{if}~ q_1 \geqslant q_2,\\
&& \label{local-dep-2} q_2(1+|q_1|) \leqslant \sqrt{1-q_2^2}\,
(1-|q_1|) ~~\text{if}~ q_1 \leqslant q_2.
\end{eqnarray}
\end{proposition}
\begin{proof}
We use equation~\eqref{2-n} with $n=2$ and apply it to all possible
output states $\varrho_{\rm out} = \mathcal{D}_{q_1} \otimes
\mathcal{D}_{q_2} [|\psi\rangle\langle\psi|]$ with
$|\psi\rangle=\sqrt{p}|00\rangle+\sqrt{1-p}|11\rangle$. It is not
hard to see that the Schmidt decomposition parameter $p=0$ or $1$
for eigenvalues $\lambda_1,\ldots,\lambda_4$ saturating
inequality~\eqref{2-n}. If $p=0,1$, then equation~\eqref{2-n} reduces
to equations~\eqref{local-dep-1}--\eqref{local-dep-2}.
\end{proof}

The area of parameters $q_1,q_2$ satisfying
equations~\eqref{local-dep-1}--\eqref{local-dep-2} is shown in
figure~\ref{figure3}. The fact that $p=0,1$ in derivation of
equations~\eqref{local-dep-1}--\eqref{local-dep-2} means that, in the
case of local depolarizing noises, the \emph{factorized} states
exhibit the most resistance to absolute separability when affected
by local depolarizing noises.

If $q_1 = q_2 = q$, then the sufficient condition
\eqref{AbsSepLocalDep} provides $\mathcal{D}_q \otimes
\mathcal{D}_q \in \text{PAS}(2|2)$ if $|q|\leqslant
\sqrt{\frac{2}{\sqrt{3}}-1} \approx 0.3933$, whereas the exact
conditions \eqref{local-dep-1}--\eqref{local-dep-2} provide
$\mathcal{D}_q \otimes \mathcal{D}_q \in \text{PAS}(2|2)$ if
$|q|\leqslant q_{\ast} \approx 0.3966$, with $q_{\ast}$ being a
solution of equation $2q_{\ast}^3 - 2q_{\ast}^2 + 3q_{\ast}-1=0$.

The boundary points of both equation~\eqref{AbsSepLocalDep} and
equations~\eqref{local-dep-1}--\eqref{local-dep-2} are $q_1 = \pm
\frac{1}{\sqrt{5}}, q_2 = \pm {\frac{1}{3}}$ and $q_1 = \pm
{\frac{1}{3}}, q_2 = \pm \frac{1}{\sqrt{5}}$. Let us recall that
$\mathcal{D}_{q_1}\otimes\mathcal{D}_{q_2}$ is entanglement
breaking if and only if $|q_1|, |q_2| \leqslant \frac{1}{3}$.
Thus, the two qubit map
$\mathcal{D}_{q_1}\otimes\mathcal{D}_{q_2}$ can be entanglement
breaking but not absolutely separating and vice versa. Thus,
$\text{PAS}(2|2)\not\subset\text{EB}$ and
$\text{EB}\not\subset\text{PAS}(2|2)$. This is related with the
fact that factorized states remain separable under the action of
local depolarizing channels, but they are the most robust states
with respect to preserving the property not to be absolutely
separable.

Moreover, $\text{PAS}(2|2)\not\subset\text{CPT}$. In fact,
$\mathcal{D}_{q_1}\otimes\mathcal{D}_{q_2}$ is completely positive
if and only if $q_1, q_2 \in [-\frac{1}{3},1]$, whereas the map
$\mathcal{D}_{0}\otimes\mathcal{D}_{-1/\sqrt{2}}$ is positive and
absolutely separating.

\begin{figure}
\centering
\includegraphics[width=8cm]{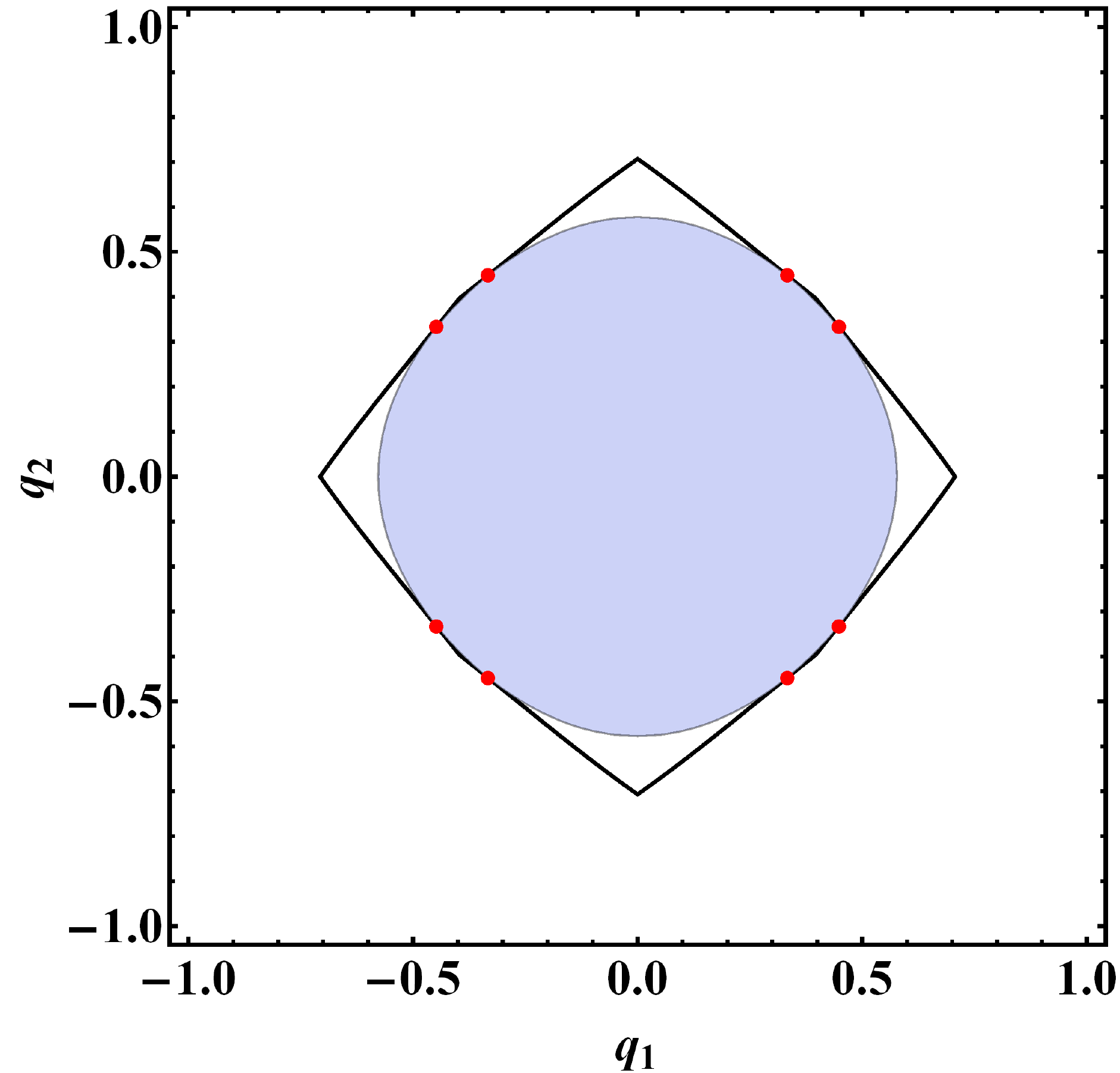}
\caption{\label{figure3} Local depolarizing two-qubit map
$\Phi_{q_1}\otimes\Phi_{q_2}$ is absolutely separating with
respect to partition $2|2$ for parameters $(q_1,q_2)$ inside the
solid line region, equations~\eqref{local-dep-1}--\eqref{local-dep-2}.
The shaded area corresponds to sufficient condition
\eqref{AbsSepLocalDep}. Points of contact between two figures are
marked by dots.}
\end{figure}

One more interesting feature is related with the fact that
$\mathcal{D}_{0}\otimes\mathcal{D}_{q_2}$ is not absolutely
separating if $q_2 > \frac{1}{\sqrt{2}}$. Physically, even though
one of the qubits is totally depolarized in the state
$\mathcal{D}_{0}\otimes\mathcal{D}_{q_2} [\varrho] = \frac{1}{2} I
\otimes \mathcal{D}_{q_2}\left[ {\rm tr}_A [\varrho] \right]$,
there exists a unitary operator $U$ (Hamilton dynamics) and a two
qubit state $\varrho$ such that $U (
\mathcal{D}_{0}\otimes\mathcal{D}_{q_2} [\varrho] ) U^{\dagger}$
is entangled with respect to $\mathcal{H}_2^A | \mathcal{H}_2^B$
if $q_2 > \frac{1}{\sqrt{2}}$. To overcome absolute separability
of the outcome, the initial state $\varrho$ should meet the
requirement $q_2^2(\lambda_1 - \lambda_2)^2 > [1+ q_2(2 \lambda_1
- 1)] [1+ q_2(2 \lambda_2 - 1)]$, where $\lambda_1,\lambda_2$ are
eigenvalues of the reduced density operator ${\rm tr}_A \varrho$.
If the state $\varrho$ satisfies this inequality, then one can
choose $U = \ket{\psi_1 \psi_1}\bra{\psi_1 \psi_1} + \ket{\psi_2
\psi_2}\bra{\psi_2 \psi_2} + \frac{1}{\sqrt{2}} \rme^{\rmi\pi/4}
(\ket{\psi_1 \psi_2}\bra{\psi_2 \psi_1} + \ket{\psi_2
\psi_1}\bra{\psi_2 \psi_1}) + \frac{1}{\sqrt{2}} \rme^{-\rmi\pi/4}
(\ket{\psi_1 \psi_2}\bra{\psi_2 \psi_1} + \ket{\psi_2
\psi_1}\bra{\psi_1 \psi_2})$, where $\ket{\psi_1},\ket{\psi_2}$
are eigenvectors of the reduced density operator ${\rm tr}_A
\varrho$.

\begin{proposition}
An $N$-qubit local depolarizing channel $\mathcal{D}_{q_1} \otimes
\ldots \otimes \mathcal{D}_{q_N}$ is absolutely separating with
respect to multipartition
$\underbrace{2|\ldots|2}_{N~\text{times}}$ if
\begin{equation}
\label{dep-N-qubits} \prod_{k=1}^N (1+q_k^2) \leqslant 1 +
\frac{54}{17}3^{-N}.
\end{equation}
\end{proposition}
\begin{proof}
The channel $\mathcal{D}_{q_1} \otimes \ldots \otimes
\mathcal{D}_{q_N}$ satisfies multiplicativity condition of the
maximum output purity~\cite{amosov-2002,king-2003}, therefore
$(\|\bigotimes_{k=1}^N \mathcal{D}_{q_k}\|_{1 \rightarrow 2})^2 =
\prod_{k=1}^N (\| \mathcal{D}_{q_k} \|_{1 \rightarrow 2})^2 =
2^{-N} \prod_{k=1}^N (1+q_k^2)$. Using
proposition~\ref{proposition-N-qubits}, we obtain
equation~\eqref{dep-N-qubits} guaranteeing the desired absolutely
separating property of $\mathcal{D}_{q_1} \otimes \ldots \otimes
\mathcal{D}_{q_N}$.
\end{proof}

\begin{example}
A local depolarizing channel $\mathcal{D}_{q}^{\otimes N}$ acting
on $N \geqslant 3$ qubits is absolutely separating with respect to
multipartition $\underbrace{2|\ldots|2}_{N~\text{times}}$ if $q\leqslant
\frac{21\sqrt{2}}{17\sqrt{N \cdot 3^N}}$.
\end{example}

Suppose each qubit experiences the same depolarizing noise, then
one can find a condition under which the resulting channel is not
absolutely separating with respect to any bipartition.

\begin{proposition}
An $N$-qubit local uniform depolarizing channel
$\mathcal{D}_{q}^{\otimes N}$ is not absolutely separating with
respect to any partition $2^k|2^{N-k}$ if
\begin{equation}
\label{dep-N-qubits-not-as} \sqrt{\frac{1+|q|}{1-|q|}} >
\frac{3+|q|}{1+|q|},
\end{equation}

\noindent or $|q| > \frac{1}{N}$.
\end{proposition}
\begin{proof}
Consider a factorized input state $(\ket{\psi}\bra{\psi})^{\otimes
N}$, then decreasingly ordered eigenvalues of
$\mathcal{D}_{q}^{\otimes N}[(\ket{\psi}\bra{\psi})^{\otimes N}]$
are $\lambda_1 = (1+|q|)^N / 2^N$, $\lambda_{2^N-2} =
\lambda_{2^N-1} = (1-|q|)^{N-1} (1+|q|) / 2^N$, and $\lambda_{2^N}
= (1-|q|)^N / 2^N$. If equation~\eqref{dep-N-qubits-not-as} is
satisfied, then the necessary condition of absolute separability
\eqref{necessary} is violated and $\mathcal{D}_{q}^{\otimes N}$ is
not absolutely separating with respect to any partition
$2^k|2^{N-k}$. Condition $|q| > \frac{1}{N}$ implies
equation~\eqref{dep-N-qubits-not-as} so it serves as a simpler
criterion of the absence of absolutely separating property.
\end{proof}

\subsection{Local unital qubit maps and channels}
\label{subsection-Pauli}

In this subsection we consider unital qubit maps
$\Upsilon:\mathcal{S}(\mathcal{H}_2) \mapsto
\mathcal{S}(\mathcal{H}_2)$, i.e. linear maps preserving maximally
mixed state, $\Upsilon[I] = I$. By a proper choice of input and
output bases the action of a general unital qubit map
reads~\cite{R-Z-W}
\begin{equation}
\label{Upsilon} \Upsilon[X] = \frac{1}{2} \sum_{j=0}^3 \lambda_j
{\rm tr}[\sigma_j X] \sigma_j,
\end{equation}

\noindent where $\sigma_0=I$ and $\{\sigma_i\}_{i=1}^3$ is a
conventional set of Pauli operators. In what follows we consider
trace preserving maps \eqref{Upsilon} with $\lambda_0 = 1$.

Consider a local unital map acting on two qubits, $\Upsilon
\otimes \Upsilon^{\prime}$. General properties of such maps are
reviewed in~\cite{filippov-rybar-ziman-2012,F-M}.

\begin{proposition}
\label{proposition-local-unital} The local unital two-qubit map
$\Upsilon \otimes \Upsilon^{\prime}$ is absolutely separating with
respect to partition $2|2$ if
\begin{equation}\label{PauliAsCondition}
\left( 1 + \max(\lambda_1^2,\lambda_2^2,\lambda_3^2) \right)
\left( 1 + \max(\lambda_1^{\prime 2},\lambda_2^{\prime
2},\lambda_3^{\prime 2}) \right) \leqslant \frac{4}{3}.
\end{equation}
\end{proposition}
\begin{proof}
The output purity ${\rm tr}\big[ (\Upsilon \otimes
\Upsilon^{\prime}[\varrho])^2 \big]$ is a convex function of
$\varrho$ and achieves its maximum $(\| \Upsilon \otimes
\Upsilon^{\prime} \|_{1 \rightarrow 2})^2$ at pure states $\varrho
= \ket{\psi}\bra{\psi}$. The Schmidt decomposition of any pure
two-qubit state $|\psi\rangle$ is $\ket{\psi} =
\sqrt{p}|\phi\otimes\chi\rangle+\sqrt{1-p}|\phi_{\perp}\otimes\chi_{\perp}\rangle$,
where $0 \leqslant p \leqslant 1$,
$\{|\phi\rangle,|\phi_{\perp}\rangle\}$ and
$\{|\chi\rangle,|\chi_{\perp}\rangle\}$ are two orthonormal bases.
We use the following parametrization by the angles
$\theta\in[0,\pi]$ and $\phi\in[0,2\pi]$:
\begin{equation}
|\phi\rangle=\begin{pmatrix}\cos({\theta/2})\exp({-\rmi\phi/2})\\\sin({\theta/2})\exp{(\rmi\phi/2)}
\end{pmatrix}, \quad
|\phi_{\perp}\rangle=\begin{pmatrix}-\sin({\theta/2})\exp({-\rmi\phi/2})\\\cos({\theta/2})\exp{(\rmi\phi/2)}
\end{pmatrix}.
\end{equation}

\noindent The basis $\{|\chi\rangle,|\chi_{\perp}\rangle\}$ is
obtained from above formulas by replacing
$|\phi\rangle\rightarrow|\chi\rangle,|\phi_{\perp}\rangle\rightarrow|\chi_{\perp}\rangle,
\theta\rightarrow\theta^{'},\phi\rightarrow\phi^{'}$. Thus, any
pure input state $\varrho = \ket{\psi}\bra{\psi}$ of two qubits
can be parameterized by 5 parameters:
$p,\theta,\phi,\theta',\phi'$. The pair $\{p,1-p\}$ is the
spectrum of reduced single-qubit density operator.

The map $\Upsilon \otimes \Upsilon^{\prime}$ transforms
$|\psi\rangle\langle\psi|$ into the operator
\begin{eqnarray}
&& \varrho_{\rm
out}(\lambda_1,\lambda_2,\lambda_3,\lambda_1',\lambda_2',\lambda_3',p,\theta,\phi,\theta',\phi')
\nonumber\\
&& = \frac{1}{4} \Big\{ I\otimes I + (\mathbf{n} \cdot
\bm{\sigma}) \otimes (\mathbf{n}^{\prime} \cdot
\bm{\sigma}^{\prime}) + (2p-1)[(\mathbf{n} \cdot\bm{\sigma})
\otimes I + I \otimes
(\mathbf{n}^{\prime} \cdot \bm{\sigma}^{\prime})] \nonumber\\
&&  + 2 \sqrt{p(1-p)} [(\mathbf{k} \cdot \bm{\sigma}) \otimes
(\mathbf{k}^{\prime} \cdot \bm{\sigma}^{\prime}) + (\mathbf{l}
\cdot \bm{\sigma}) \otimes (\mathbf{l}^{\prime} \cdot
{\bm{\sigma}}^{\prime})] \Big\},
\end{eqnarray}

\noindent where $(\mathbf{n} \cdot \bm{\sigma}) = n_1 \sigma_1 +
n_2 \sigma_2 + n_3 \sigma_3$ and vectors
$\mathbf{n},\mathbf{k},\mathbf{l}\in \mathbb{R}^3$ are expressed
through parameters of map $\Upsilon$ by formulas
\begin{eqnarray}
&&\mathbf{n}=(\lambda_1\cos\phi\sin\theta,\lambda_2\sin\phi\sin\theta,\lambda_3\cos\theta),\\
&&\mathbf{k}=(-\lambda_1\cos\phi\cos\theta,-\lambda_2\sin\phi\cos\theta,\lambda_3\sin\theta),\\
&&\mathbf{l}=(\lambda_1\sin\phi,-\lambda_2\cos\phi,0).
\end{eqnarray}

\noindent The vectors $\mathbf{n}',\mathbf{k}',\mathbf{l}'$  are
obtained from $\mathbf{n},\mathbf{k},\mathbf{l}$, respectively, by
replacing $\mathbf{\lambda}\rightarrow\mathbf{\lambda}',
\theta\rightarrow\theta', \phi\rightarrow\phi'$. Maximizing the
output purity $\text{tr}[\varrho_{\rm out}^2]$ over $p\in[0,1]$,
we get
\begin{eqnarray}
\label{max-p-pauli}  &&\max_{p\in[0,1]} {\rm tr}\big[ \big(
\varrho_{\rm out}
(\lambda_1,\lambda_2,\lambda_3,\lambda_1',\lambda_2',\lambda_3',p,\theta,\phi,\theta',\phi')
\big)^2 \big] \nonumber \\ && =  \frac{1}{4} \left(
1+\lambda_3^2\cos^2\theta+(\lambda_1^2\cos^2\phi+\lambda_2^2\sin^2\phi)\sin^2\theta
\right) \nonumber \\ && \times \left( 1+\lambda_3^{\prime 2}\cos^2
\theta^{\prime} + (\lambda_1^{\prime 2}\cos^2 \phi^{\prime} +
\lambda_2^{\prime 2} \sin^2 \phi^{\prime}) \sin^2\theta^{\prime}
\right),
\end{eqnarray}

\noindent which is achieved at factorized states ($p=0$ or $p=1$).
Maximizing equation~\eqref{max-p-pauli} over angles
$\theta,\phi,\theta^{\prime},\phi^{\prime}$, we get
\begin{equation}\label{PauliAsCondition}
( \| \Upsilon \otimes \Upsilon^{\prime} \|_{1 \rightarrow 2} )^2
 = \frac{1}{4} \left( 1 +
\max(\lambda_1^2,\lambda_2^2,\lambda_3^2) \right) \left( 1 +
\max(\lambda_1^{\prime 2},\lambda_2^{\prime 2},\lambda_3^{\prime
2}) \right).
\end{equation}

\noindent By proposition~\ref{proposition-ball}, $\Upsilon \otimes
\Upsilon^{\prime}$ is absolutely separating with respect to
partition $2|2$ if $(\| \Upsilon \otimes \Upsilon^{\prime} \|_{1
\rightarrow 2})^2 \leqslant \frac{1}{3}$, which implies
equation~\eqref{PauliAsCondition}.
\end{proof}

As in the case of local depolarizing maps, pure factorized states
are the most resistant to absolute separability under action of
$\Upsilon \otimes \Upsilon^{\prime}$. If the map $\Upsilon \otimes
\Upsilon^{\prime}$ were completely positive, one could use the
multiplicativity condition for calculation of the maximal output
purity~\cite{king-2-2002}. However, in our case the map $\Upsilon
\otimes \Upsilon^{\prime}$ is not necessarily completely positive.

The map $\Upsilon \otimes \Upsilon$ is absolutely separating with
respect to partition $2|2$ if
$\max(\lambda_1^2,\lambda_2^2,\lambda_3^2) \leqslant
\frac{2}{\sqrt{3}} - 1$. The area of parameters
$\lambda_1,\lambda_2,\lambda_3$ satisfying this inequality is
depicted in figure~\ref{figure4}. Clearly, $\Upsilon \otimes
\Upsilon$ may be absolutely separating even if it is not
completely positive.

\begin{figure}
\centering
\includegraphics[width=8cm]{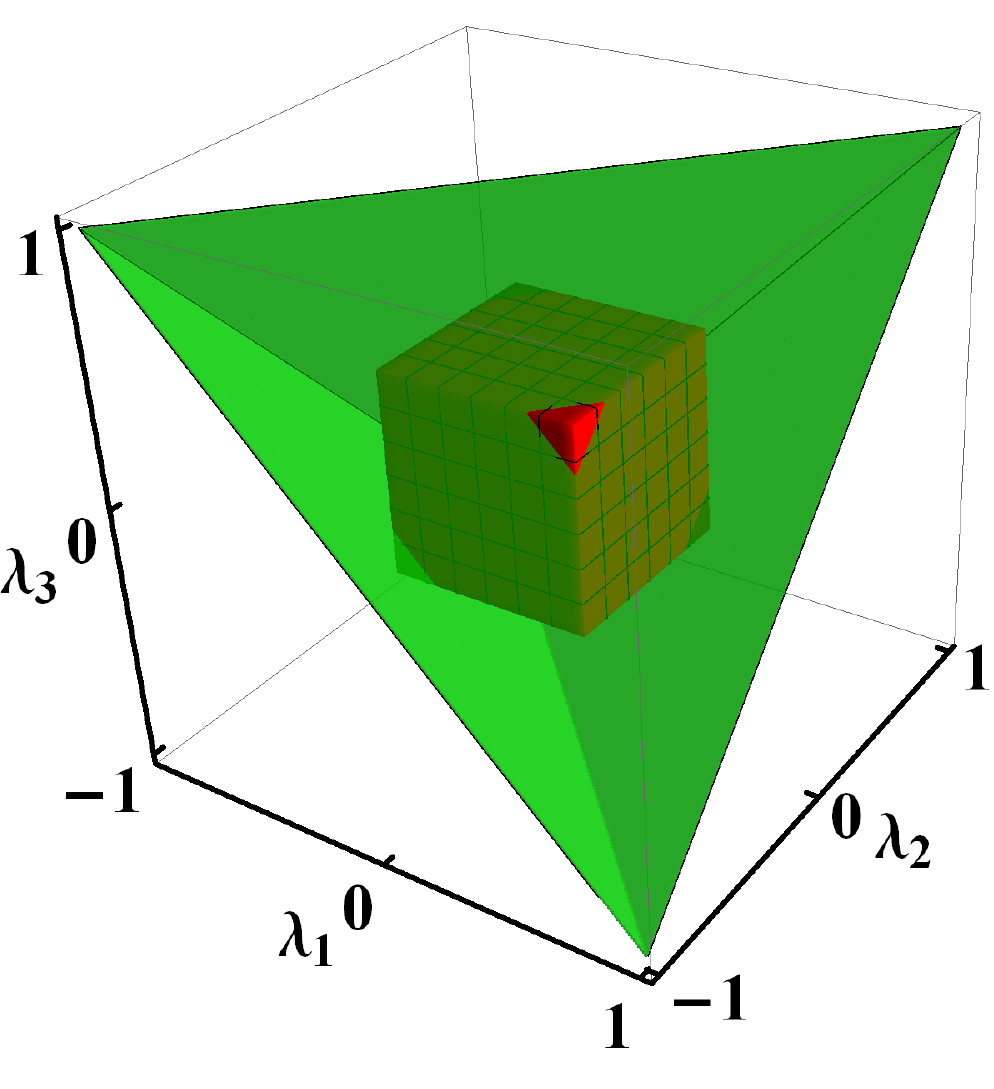}
\caption{\label{figure4} Parameters
$\lambda_1,\lambda_2,\lambda_3$ of the Pauli map $\Upsilon$, where
$\Upsilon \otimes \Upsilon$ is completely positive (green
tetrahedron) and absolutely separating with respect to partition
$2|2$ by proposition~\ref{proposition-local-unital} (red cube).}
\end{figure}

\begin{proposition}
An $N$-qubit local unital channel $\Upsilon^{(1)} \otimes \ldots
\otimes \Upsilon^{(N)}$ is absolutely separating with respect to
multipartition $\underbrace{2|\ldots|2}_{N~\text{times}}$ if
\begin{equation}
\prod_{k=1}^N \left\{ 1 +
[\max(|\lambda_1^{(k)}|,|\lambda_2^{(k)}|,|\lambda_3^{(k)}|)]^2
\right\} \leqslant 1 + \frac{54}{17}3^{-N}.
\end{equation}
\end{proposition}
\begin{proof}
The proof follows from the multiplicativity of the maximum output
purity~\cite{king-2-2002} and
proposition~\ref{proposition-N-qubits}.
\end{proof}

\subsection{Generalized Pauli channels}
\label{subsection-generalized-Pauli}

The maps considered in previous subsections were local. Let us
consider a particular family of non-local maps called generalized
Pauli channels or Pauli diagonal channels constant on
axes~\cite{NatRus}. Suppose an $mn$-dimensional Hilbert space
$\mathcal{H}_{mn}$ and a collection
$\mathcal{B}_J=\{|\psi_k^J\rangle\}_{k=1}^{mn}$ of orthonormal
bases in $\mathcal{H}_{mn}$. For simplicity denote $d=mn$ and
define the operators
\begin{equation}
W_J = \sum_{k=1}^{d} \omega^k |\psi_k^J\rangle\langle\psi_k^J|,
\qquad J=1,2,\dots,d+1,
\end{equation}

\noindent where $\omega=\rme^{\rmi 2\pi / d}$. If $d$ is a power
of a prime number, then there exist $d+1$ mutually unbiased
bases~\cite{wootters-1989}. The corresponding $d^2-1$ unitary
operators $\{W_J^m\}_{m=1,\dots,d-1,J=1,\dots,d+1}$ satisfy the
orthogonality condition $\text{tr}[ (W^j_J)^{\dag} W^k_K ] = d
\delta_{JK} \delta_{jk}$ and, hence, form an orthonormal basis for
the subspace of traceless matrices.

A generalized Pauli channel $\Phi$ acts on $\varrho\in
\mathcal{S}(\mathcal{H}_d)$ as follows:
\begin{equation}\label{Phi}
\Phi[\varrho]=\frac{(d-1)s+1}{d} \varrho +
\frac{1}{d}\sum_{J=1}^{d+1}\sum_{j=1}^{d-1}t_JW_J^j \varrho
(W_J^j)^{\dag}.
\end{equation}

\noindent Conditions
\begin{equation}
s+\sum_{J=1}^{d+1} t_J = 1, \quad t_J \geqslant 0, \quad s
\geqslant - \frac{1}{d-1}
\end{equation}

\noindent on parameters $s,t_1,\ldots,t_{d+1}$ ensure that $\Phi$
is trace preserving and completely positive ($\Phi$ is a quantum
channel). To analyse absolutely separating properties we use
Theorem 27 in~\cite{NatRus}: the maximal output purity of $\Phi$
is achieved with an axis state, i.e. there exist $n$ and $J$ such
that $(\|\Phi\|_{1\rightarrow 2})^2 = {\rm
tr}\left[(\Phi[|\psi_n^{J}\rangle\langle\psi_n^J|])^2\right]$. On
the other side, action of the generalized Pauli channel on an axis
state $|\psi_n^{J}\rangle\langle\psi_n^J|$ reads
\begin{equation}
\label{axis-on-axis} \Phi[|\psi_n^{J}\rangle\langle\psi_n^J|] =
(1-s-t_J) \frac{1}{d}I + (s+t_J) |\psi_n^J\rangle\langle\psi_n^J|,
\end{equation}

\noindent whose purity equals $[1+(d-1)(s+t_J)^2]/d$. Thus, if the
obtained purity is less or equal to $(d-1)^{-1}$ for all $J$, then
by proposition~\ref{proposition-ball} $\Phi$ is absolutely
separating with respect to $m|n$. To conclude, a generalized Pauli
channel $\Phi:\mathcal{S}(\mathcal{H}_{mn})\mapsto
\mathcal{S}(\mathcal{H}_{mn})$ is absolutely separating with
respect to partition $m|n$ if $|s+t_J| \leqslant (mn-1)^{-1}$ for
all $J=1,\ldots,mn+1$.

\subsection{Combination of tracing, identity, and transposition maps}
\label{subsection-ctit}

Let us consider a two-parametric family of positive maps
$\Phi_{\alpha\beta}:\mathcal{S}(\mathcal{H}_d) \mapsto
\mathcal{S}(\mathcal{H}_d)$ representing linear combinations of
tracing map, identity transformation, and transposition $\top$ in
a fixed orthonormal basis:
\begin{equation}
\label{ctit} \Phi_{\alpha\beta} [X] = \frac{1}{d + \alpha + \beta}
\left( {\rm tr}[X] \, I + \alpha X + \beta X^{\top} \right)
\end{equation}

\noindent with real parameters $\alpha$ and $\beta$ satisfying
inequalities $1+\alpha \geqslant 0$, $1+\beta \geqslant 0$, and
$1+\alpha+\beta \geqslant 0$ (guaranteeing $\Phi_{\alpha\beta}$ is
positive). Note that $\Phi_{\alpha\beta}$ is trace preserving.
Equation~\eqref{ctit} reduces to the depolarizing map if $\beta =
0$ and to the Werner-Holevo channel~\cite{werner-holevo-2002} if
$\alpha = 0$ and $\beta = -1$. A direct calculation of the
Choi-Jamio{\l}kowski operator~\cite{choi-1975,jamiolkowski-1972}
shows that $\Phi_{\alpha\beta}$ is completely positive if $\alpha
\geqslant - \frac{1}{d}$ and $-(1+d \alpha) \leqslant \beta
\leqslant 1$.

Suppose $\mathcal{H}_d = \mathcal{H}_m \otimes \mathcal{H}_n$,
then we can explore the absolute separability of the output
$\Phi_{\alpha\beta} [\varrho]$ with respect to partition $m|n$.

\begin{proposition}
\label{proposition-ctit}
$\Phi_{\alpha\beta}:\mathcal{S}(\mathcal{H}_{mn}) \mapsto
\mathcal{S}(\mathcal{H}_{mn})$ is absolutely separating with
respect to partition $m|n$ if
\begin{equation}
\label{stripe} -1 \leqslant \alpha + \beta \leqslant
\frac{mn}{mn-2}
\end{equation}
\noindent and
\begin{equation}
\label{ellipse}
(\alpha-\beta)^2 + \frac{mn-3}{mn-1}\left(
\alpha+\beta-\frac{2}{mn-3} \right)^2 \leqslant
\frac{2(mn-2)}{mn-3}.
\end{equation}
\end{proposition}
\begin{proof}
Since ${\rm tr}[\varrho] = {\rm tr}[\varrho^{\top}]=1$ for a
density matrix $\varrho$ and ${\rm tr}[\varrho^2] = {\rm
tr}[(\varrho^{\top})^2]$, the output purity of the map
$\Phi_{\alpha\beta}$ reads
$(d+\alpha+\beta)^{-2}\{d+2(\alpha+\beta)+2\alpha\beta {\rm
tr}[\varrho \varrho^{\top}] + (\alpha^2+\beta^2){\rm
tr}[\varrho^2]\}$. If $\alpha\beta \geqslant 0$, then the output
purity is maximal when ${\rm tr}[\varrho \varrho^{\top}]=1$ and
${\rm tr}[\varrho^2]=1$. Substituting the obtained value of the
maximum output purity in equation~\eqref{ball}, we get
equation~\eqref{stripe}. If $\alpha\beta < 0$, then the output purity
is maximal when ${\rm tr}[\varrho \varrho^{\top}]=0$ and ${\rm
tr}[\varrho^2]=1$. If this is the case, equation~\eqref{ball} results
in equation~\eqref{ellipse}. Combining two criteria, we see that if
both conditions \eqref{stripe} and \eqref{ellipse} are fulfilled,
then $\Phi_{\alpha\beta}$ is absolutely separating with respect to
partition $m|n$ by proposition~\ref{proposition-ball}.
\end{proof}

The region of parameters $\alpha,\beta$ satisfying
equations~\eqref{stripe}--\eqref{ellipse} is the intersection of a
stripe and an ellipse depicted in figure~\ref{figure4}.

According to proposition~\ref{proposition-ctit} the Werner-Holevo
channel $\Phi_{0,-1}:\mathcal{S}(\mathcal{H}_{mn}) \mapsto
\mathcal{S}(\mathcal{H}_{mn})$ is absolutely separating with
respect to partition $m|n$ for all $m,n = 2,3,\ldots$.

If $\beta=0$, then $\Phi_{\alpha,0}:\mathcal{S}(\mathcal{H}_{mn})
\mapsto \mathcal{S}(\mathcal{H}_{mn})$ is absolutely separating
with respect to partition $m|n$ when $-1 \leqslant \alpha
\leqslant \frac{mn}{mn-2}$, which corresponds to the global
depolarizing map $\mathcal{D}_{q}$ with $|q|\leqslant
\frac{1}{mn-1}$.

If $m=2$, one can use a necessary and sufficient condition
\eqref{2-n} of absolute separability with respect to partition
$2|n$ and apply it to the map \eqref{ctit}.

\begin{proposition}
\label{proposition-ctit-2-n}
$\Phi_{\alpha\beta}:\mathcal{S}(\mathcal{H}_{2n}) \mapsto
\mathcal{S}(\mathcal{H}_{2n})$ is absolutely separating with
respect to partition $2|n$ if and only if (i) $\alpha,\beta
\geqslant 0$ and $\alpha + \beta \leqslant 2$; (ii) $\alpha
\geqslant 0$ and $\alpha^2-4 \leqslant 4 \beta < 0$; (iii) $\beta
\geqslant 0$ and $\beta^2-4 \leqslant 4 \alpha < 0$; (iv)
$\alpha,\beta < 0$ and $\alpha + \beta \geqslant -1$.
\end{proposition}
\begin{proof}
Since the state space is convex and the map $\Phi_{\alpha\beta}$
is linear, it is enough to check absolute separability of the
output $\Phi_{\alpha\beta}[\ket{\psi}\bra{\psi}]$ for pure states
$\ket{\psi}\in\mathcal{H}_{2n}$ only. Transposition
$(\ket{\psi}\bra{\psi})^{\top} =
\ket{\overline{\psi}}\bra{\overline{\psi}}$ is equivalent to
complex conjugation $\ket{\psi} \mapsto \ket{\overline{\psi}}$ in
a fixed basis. Eigenvalues of the operator
$I+\alpha\ket{\psi}\bra{\psi} + \beta
\ket{\overline{\psi}}\bra{\overline{\psi}}$ are $1$ with
degeneracy $2n-2$ and $1 + \frac{1}{2}(\alpha + \beta) \pm
\frac{1}{2} \sqrt{(\alpha - \beta)^2 + 4 \alpha \beta
|\ip{\psi}{\overline{\psi}}|^2}$. Since $0 \leqslant
|\ip{\psi}{\overline{\psi}}|^2 \leqslant 1$, the largest possible
eigenvalue is $\lambda_1 = 1+\frac{1}{2}(\alpha +
\beta)+\frac{1}{2} \max(|\alpha+\beta|,|\alpha-\beta|)$, while the
smallest eigenvalue is $\lambda_{2n} = 1+\frac{1}{2}(\alpha +
\beta) - \frac{1}{2} \max(|\alpha+\beta|,|\alpha-\beta|)$, with
other eigenvalues being equal to 1. Substituting such a spectrum
in equation~\eqref{2-n}, we get conditions (i)--(iv).
\end{proof}

We depict the region of parameters $\alpha,\beta$ corresponding to
$\Phi_{\alpha,\beta} \in \text{PAS}(2|4)$ in figure~\ref{figure5}.

\begin{figure}
\centering
\includegraphics[width=8cm]{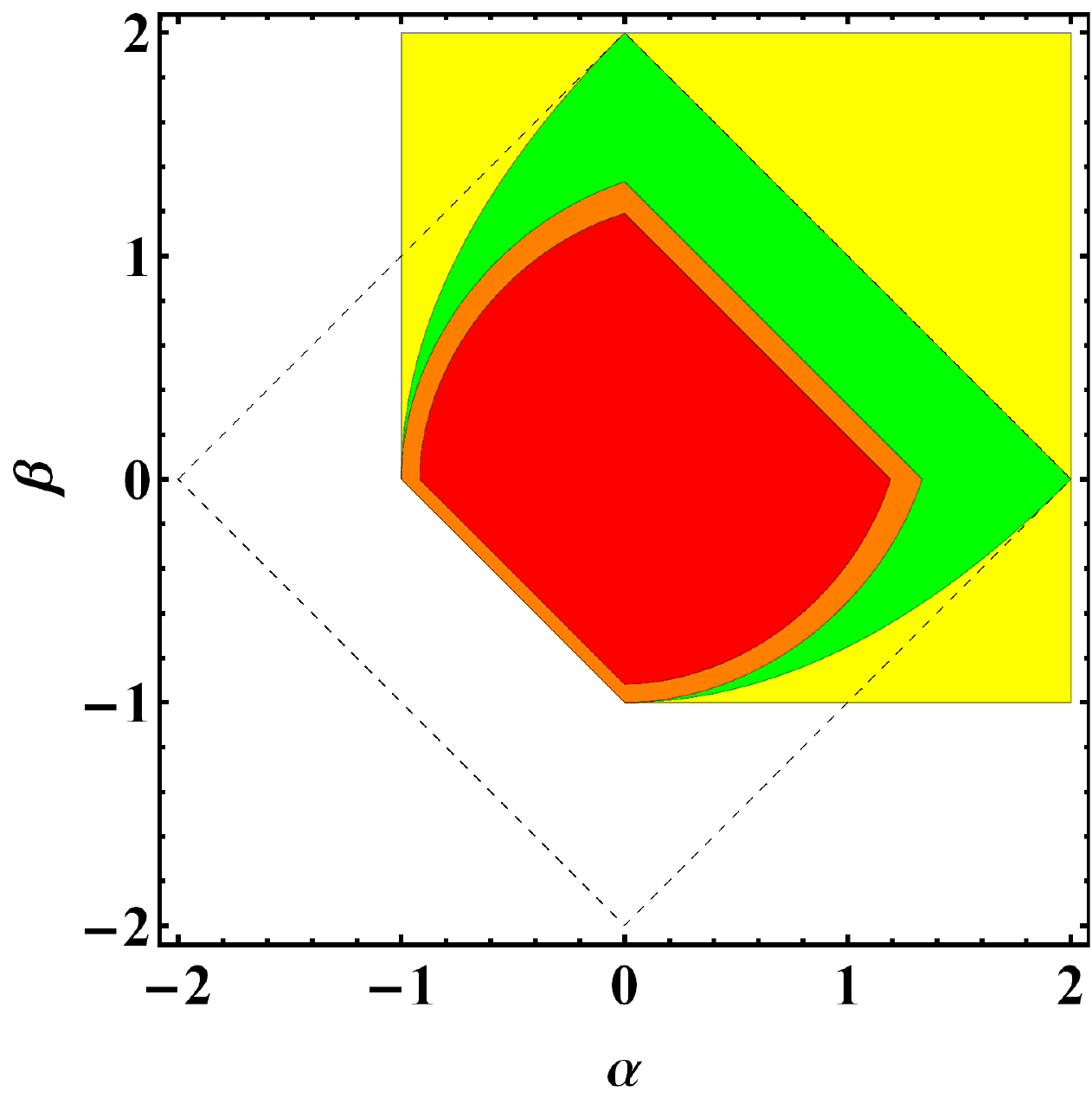}
\caption{\label{figure5} Nested structure of two-parametric maps
$\Phi_{\alpha\beta}: \mathcal{S}(\mathcal{H}_8) \mapsto
\mathcal{S}(\mathcal{H}_8)$. Shaded areas from outer to inner
ones: $\Phi_{\alpha\beta}$ is positive, $\Phi_{\alpha\beta}$ is
absolutely separating with respect to partition $2|4$ by necessary
and sufficient criterion of
proposition~\ref{proposition-ctit-2-n}, sufficient condition of
absolutely separating property with respect to partition $2|4$ by
proposition~\ref{proposition-ctit}, $\Phi_{\alpha\beta}$ is
absolutely separating with respect to multipartition $2|2|2$ by
equation~\eqref{ctit-2-2-2}. The dashed square represents a necessary
condition of absolute separating property with respect to any
bipartition $m|n$.}
\end{figure}

Let us consider a necessary condition of the absolutely separating
property of $\Phi_{\alpha\beta}$.

\begin{proposition}
\label{proposition-ctit-necessary} Suppose
$\Phi_{\alpha\beta}:\mathcal{S}(\mathcal{H}_{mn}) \mapsto
\mathcal{S}(\mathcal{H}_{mn})$ is absolutely separating with
respect to partition $m|n$, then
$\max(|\alpha+\beta|,|\alpha-\beta|) \leqslant 2$.
\end{proposition}
\begin{proof}
If $\Phi_{\alpha\beta} \in \text{PAS}(m|n)$, then
equation~\eqref{necessary} is to be satisfied for the spectrum of
states $\Phi_{\alpha\beta}[\varrho]$. Let us recall that the
largest and smallest eigenvalues of the operator
$(mn+\alpha+\beta)\Phi_{\alpha\beta}[\ket{\psi}\bra{\psi}]$ read
$\lambda_1 = 1+\frac{1}{2}(\alpha + \beta)+\frac{1}{2}
\max(|\alpha+\beta|,|\alpha-\beta|)$ and $\lambda_{2n} =
1+\frac{1}{2}(\alpha + \beta) - \frac{1}{2}
\max(|\alpha+\beta|,|\alpha-\beta|)$, respectively. Eigenvalues
$\lambda_{2} = \ldots = \lambda_{mn-1} = 1$. Substituting
$\lambda_1$, $\lambda_{mn-2}$, $\lambda_{mn-1}$, $\lambda_{mn}$
into equation~\eqref{necessary}, we get
$\max(|\alpha+\beta|,|\alpha-\beta|) \leqslant 2$.
\end{proof}

The obtained necessary condition does not depend on $m$ and $n$
and is universal for the maps $\Phi_{\alpha\beta}$.

Finally, by proposition~\ref{proposition-N-qubits},
$\Phi_{\alpha\beta}:\mathcal{S}(\mathcal{H}_{2^N}) \mapsto
\mathcal{S}(\mathcal{H}_{2^N})$ is absolutely separating with
respect to $N$-partition $2|\ldots|2$ if
\begin{equation}
\label{ctit-2-2-2}  2^N + 2(\alpha + \beta) + |\alpha\beta| +
\alpha\beta + \alpha^2 +
\beta^2 \leqslant \frac{(2^N+\alpha+\beta)^2}{2^N} \left( 1 +
\frac{54}{17} 3^{-N} \right).
\end{equation}

\noindent As an example we illustrate the region of parameters
$\alpha,\beta$, where $\Phi_{\alpha\beta}$ is $\text{PAS}(2|2|2)$,
see figure~\ref{figure5}.

\subsection{Bipartite depolarizing channel}
\label{subsection-bipartite-depolarizing}

Suppose a bipartite physical system whose parts are far apart from
each other, then the interaction with individual environments
leads to local noises, for instance, local depolarization
considered in section~\ref{subsection-depolarizing}. In contrast,
if the system is compact enough to interact with the common
environment as a whole, the global noise takes place. As an
example, the global depolarization is a map $\Phi_{\alpha,0}$
considered in section~\ref{subsection-ctit}. In general, two parts
of a composite system $AB$ can be separated in such a way that
both global and local noises affect it. Combination of global and
local depolarizing maps results in the map
$\Phi:\mathcal{S}(\mathcal{H}_{m}^A \otimes \mathcal{H}_{n}^B)
\mapsto \mathcal{S}(\mathcal{H}_{m}^A \otimes \mathcal{H}_{n}^B)$
\begin{equation}
\label{Phi-alpha-beta-gamma} \Phi_{\alpha\beta\gamma}[X] =
\frac{I_{mn} {\rm tr}[X] + \alpha I_{m} \otimes {\rm tr}_A[X] +
\beta {\rm tr}_B[X] \otimes I_n + \gamma X}{mn + \alpha m + \beta
n + \gamma},
\end{equation}

\noindent whose positivity and entanglement annihilating
properties were explored in the paper~\cite{lami-huber-2016}.

The output purity reads
\begin{eqnarray}
&& \fl {\rm tr}\left[ \left(\Phi_{\alpha\beta\gamma}[\varrho]
\right)^2
\right] = (mn + \alpha m + \beta n + \gamma)^{-2}\nonumber\\
&& \fl \times \!\big\{\! mn \!+ \!2(\alpha m \!+ \!\beta n \!+\!
\alpha \beta \!+\! \gamma)\! + \! \gamma^2 {\rm tr}[\varrho^2]  \!
+\! (\alpha^2 m\! +\! 2 \alpha \gamma) {\rm tr}[\varrho_B^2] \!+\!
(\beta^2 n \! + \! 2\beta\gamma) {\rm tr}[\varrho_A^2] \big\},
\end{eqnarray}

\noindent where $\varrho_A = {\rm tr}_B \varrho$ and $\varrho_B =
{\rm tr}_A \varrho$. Also, we have taken into account the fact
that ${\rm tr}[\varrho (I_m \otimes \varrho_B)] = {\rm
tr}[\varrho_B^2]$ and ${\rm tr}[\varrho (\varrho_A \otimes I_n)] =
{\rm tr}[\varrho_A^2]$.

Let us recall that $\Phi_{\alpha\beta\gamma}$ is absolutely
separating with respect to partition $m|n$ if and only if
$\Phi_{\alpha\beta\gamma}[\ket{\psi}\bra{\psi}] \in
\mathcal{A}(m|n)$ for all pure states $\ket{\psi} \in
\mathcal{S}(\mathcal{H}_{mn})$. It means that we can restrict
ourselves to the analysis of pure input states $\varrho =
\ket{\psi}\bra{\psi}$ satisfying ${\rm tr}[\varrho^2] = 1$. On the
other hand, reduced density operators $\varrho_A$ and $\varrho_B$
have the same spectra if $\varrho$ is pure, therefore ${\rm
tr}[\varrho_A^2] = {\rm tr}[\varrho_B^2] = \mu \in
[\frac{1}{\min(m,n)},1]$. Thus,
\begin{eqnarray}
&& \fl {\rm tr}\left[
\left(\Phi_{\alpha\beta\gamma}[\ket{\psi}\bra{\psi}] \right)^2
\right] = (mn + \alpha m + \beta n + \gamma)^{-2} \nonumber\\
&& \fl \times \big\{ mn + 2(\alpha m + \beta n + \alpha \beta +
\gamma) + \gamma^2  + [\alpha^2 m + \beta^2 n + 2 \gamma (\alpha +
\beta)] \mu \big\}. \label{abc-output-purity}
\end{eqnarray}

If $\alpha^2 m + \beta^2 n + 2 \gamma (\alpha + \beta) > 0$, then
expression~\eqref{abc-output-purity} achieves its maximum at a
factorized state $\ket{\psi} = \ket{\phi}_A \otimes \ket{\chi}_B$,
when $\mu = 1$. If $\alpha^2 m + \beta^2 n + 2 \gamma (\alpha +
\beta) < 0$, then expression~\eqref{abc-output-purity} achieves
its maximum at the maximally entangled state $\ket{\psi} =
\frac{1}{\sqrt{\min(m,n)}} \sum_{i=1}^{\min(m,n)} \ket{i} \otimes
\ket{i}$ with $\mu = \frac{1}{\min(m,n)}$.

Using the explicit form of the maximum output purity
\eqref{abc-output-purity} and proposition~\ref{proposition-ball},
we get the following result.

\begin{proposition}
\label{proposition-abc-sufficient}
$\Phi_{\alpha\beta\gamma}:\mathcal{S}(\mathcal{H}_{mn}) \mapsto
\mathcal{S}(\mathcal{H}_{mn})$ is $\text{PAS}(m|n)$ if $\alpha^2 m
+ \beta^2 n + 2 \gamma (\alpha + \beta) \geqslant 0$ and
\begin{equation}
(\alpha+\beta+\gamma)^2 + \alpha^2(m-1) + \beta^2 (n-1) - 1
\leqslant \frac{(\alpha m + \beta n + \gamma+1)^2}{mn-1},
\end{equation}

\noindent or $\alpha^2 m + \beta^2 n + 2 \gamma (\alpha + \beta)
\leqslant 0$ and
\begin{equation}
\label{bi-dep-max-ent} \gamma^2+2\alpha\beta - 1 + \frac{\alpha^2
m + \beta^2 n + 2\gamma(\alpha + \beta)}{\min(m,n)} \leqslant
\frac{(\alpha m + \beta n + \gamma+1)^2}{mn-1} .
\end{equation}
\end{proposition}

If $m=n$, then equation~\eqref{bi-dep-max-ent} reduces to $(n^2 -
1)|\gamma| \leqslant |\gamma + n(\alpha + \beta) + n^2|$.

\begin{proposition}
\label{proposition-abc-necessary} Suppose
$\Phi_{\alpha\beta\gamma} \in \text{PAS}(m|n)$, then the
decreasingly ordered vectors $\boldsymbol{\lambda}$ of the form
\begin{eqnarray}
&& \label{spectrum-factorized} (1+\alpha+\beta+\gamma,
\underbrace{1+\alpha, \ldots}_{\scriptsize{\begin{array}{c}
  m-1 \\
  \text{times} \\
\end{array}}},
\underbrace{1+\beta, \ldots}_{\scriptsize{\begin{array}{c}
  n-1 \\
  \text{times} \\
\end{array}}},
\underbrace{1, \ldots}_{\scriptsize{\begin{array}{c}
  (m-1) \\
  \times (n-1) \\
  \text{times} \\
\end{array}}})^{\downarrow}, \\
&& \label{spectrum-entangled}
(1+\tfrac{\alpha+\beta}{\min(m,n)}+\gamma,
\underbrace{1+\tfrac{\alpha+\beta}{\min(m,n)},
\ldots}_{\scriptsize{\begin{array}{c}
  [\min(m,n)]^2-1 \\
  \text{times} \\
\end{array}}}, \underbrace{1,
\ldots}_{\scriptsize{\begin{array}{c}
  mn \\
  -[\min(m,n)]^2 \\
  \text{times} \\
\end{array}}})^{\downarrow} \qquad
\end{eqnarray}

\noindent must satisfy $\lambda_1^{\downarrow} \leqslant
\lambda_{mn-2}^{\downarrow} + \lambda_{mn-1}^{\downarrow} +
\lambda_{mn}^{\downarrow}$ and $\lambda_{mn}^{\downarrow}
\geqslant 0$.
\end{proposition}
\begin{proof}
Equations~\eqref{spectrum-factorized} and \eqref{spectrum-entangled}
are nothing else but the spectra of the operator $I_{mn} {\rm
tr}[\varrho] + \alpha I_{m} \otimes {\rm tr}_A[\varrho] + \beta
{\rm tr}_B[\varrho] \otimes I_n + \gamma \varrho$ for the
factorized state $\varrho = \ket{\varphi}\bra{\varphi} \otimes
\ket{\chi}\bra{\chi} \in \mathcal{S}(\mathcal{H}_m|\mathcal{H}_n)$
and the maximally entangled state $\varrho = \frac{1}{\min(m,n)}
\sum_{i,j=1}^{\min(m,n)} \ket{i}\bra{j} \otimes \ket{i}\bra{j} \in
\mathcal{S}(\mathcal{H}_{mn})$, respectively. Since
$\Phi_{\alpha\beta\gamma}[\varrho] \in \mathcal{A}(m|n)$, the
spectra of $\Phi_{\alpha\beta\gamma}[\varrho]$ must satisfy
equation~\eqref{necessary}, and so do the spectra
\eqref{spectrum-factorized}--\eqref{spectrum-entangled} in view of
the relation~\eqref{Phi-alpha-beta-gamma}. Requirement
$\lambda_{mn}^{\downarrow} \geqslant 0$ is merely the positivity
requirement for output density operators.
\end{proof}

\begin{figure}
\centering
\includegraphics[width=8cm]{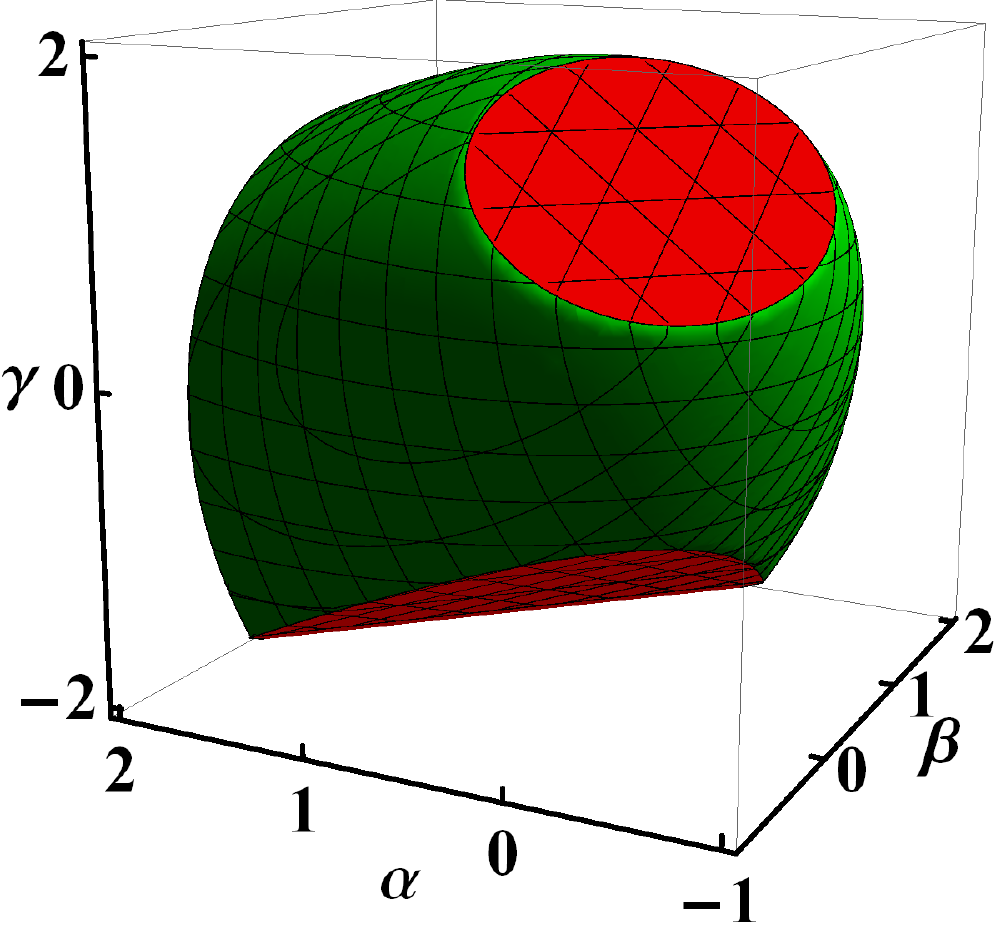}
\caption{\label{figure6} Region of parameters, where
$\Phi_{\alpha\beta\gamma}:\mathcal{S}(\mathcal{H}_4) \mapsto
\mathcal{S}(\mathcal{H}_4)$ is absolutely separating with respect
to partition $2|2$. Plane sections (red) corresponds to maximally
entangled input states, convex surface (green) corresponds to
factorized input states.}
\end{figure}

\begin{example}
Let $m=n=2$. Parameters $\alpha,\beta,\gamma$ satisfying both
propositions~\ref{proposition-abc-sufficient} and
\ref{proposition-abc-necessary} are depicted in
figure~\ref{figure6}. Note that these sufficient and (separately)
necessary conditions do coincide for parameters
$\alpha,\beta,\gamma$ in the vicinity of the upper plane section
in figure~\ref{figure6}, with the maximally entangled state being
the most resistant to absolute separability. Lower plane section
in figure~\ref{figure6} corresponds to positivity condition
$\lambda_{mn}^{\downarrow} \geqslant 0$.
\end{example}

\begin{example}
Let $m=n=3$. Parameters $\alpha,\beta,\gamma$ satisfying
proposition~\ref{proposition-abc-sufficient} are depicted by a
shaded body in figure~\ref{figure7}. Plane sections correspond to
maximally entangled states (red), and convex surface (green)
corresponds to factorized input states. A polyhedron in
figure~\ref{figure7} corresponds to
proposition~\ref{proposition-abc-necessary}. The upper and lower
faces of that polyhedron correspond to maximally entangled initial
states, and all other faces correspond to factorized input states.
\end{example}

\begin{figure}
\centering
\includegraphics[width=8cm]{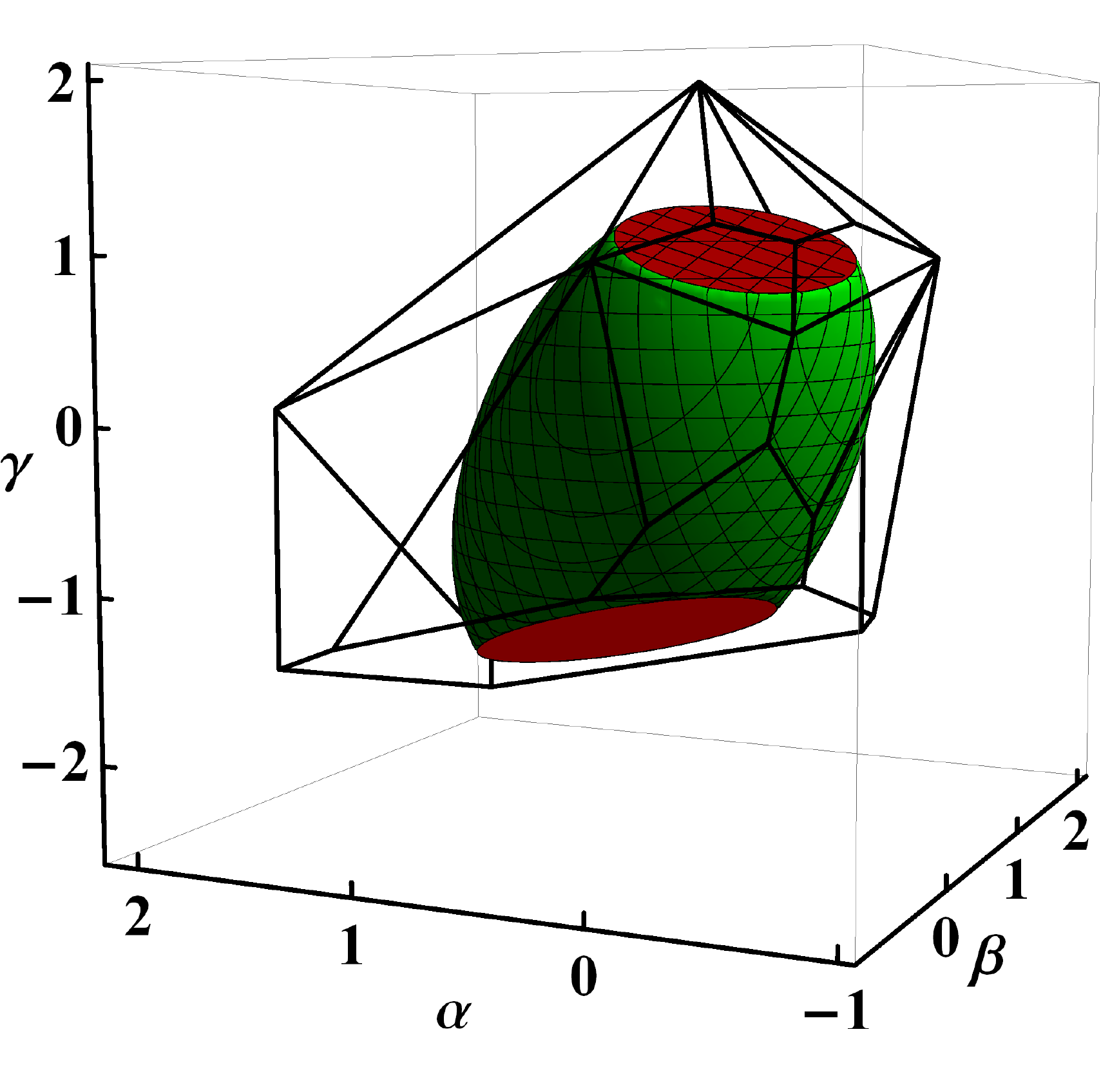}
\caption{\label{figure7} The map
$\Phi_{\alpha\beta\gamma}:\mathcal{S}(\mathcal{H}_9) \mapsto
\mathcal{S}(\mathcal{H}_9)$ is absolutely separating with respect
to partition $3|3$ by proposition~\ref{proposition-abc-sufficient}
for parameters $\alpha,\beta,\gamma$ inside the colored region
(sufficient condition). Parameters $\alpha,\beta,\gamma$ must
belong to the polyhedron for
$\Phi_{\alpha\beta\gamma}:\mathcal{S}(\mathcal{H}_9) \mapsto
\mathcal{S}(\mathcal{H}_9)$ to be absolutely separating with
respect to partition $3|3$ (necessary condition).}
\end{figure}

\section{Discussion of state robustness}
\label{section-discussion}

Let us now summarize observations of the state resistance to
absolute separability.

Suppose a dynamical process $\Phi_t$ described by a local
depolarizing or a unital $N$-qubit channel, $\bigotimes_{k=1}^N
\mathcal{D}_{q_k}$ and $\bigotimes_{k=1}^N \Upsilon^{(k)}$, with
monotonically decreasing parameters $q_k(t)$ or
$\lambda_i^{(k)}(t)$, $q_k(0)=\lambda_i^{(k)}(0) = 1$. Then the
analysis of sections~\ref{subsection-depolarizing} and
\ref{subsection-Pauli} shows that a properly chosen factorized
pure initial state $\varrho = \bigotimes_{k=1}^N
\ket{\psi_k}\bra{\psi_k}$ affected by the dynamical map $\Phi_t$
remains not absolutely separable for the longer time $t$ as
compared to initially entangled states. The matter is that
factorized states exhibit a less decrease of purity in this case
as compared to entangled states whose purity decreases faster due
to the destruction of correlations.

State robustness is irrelevant to the initial degree of state
entanglement in the case of evolution under a linear combination
of global tracing, identity, and transposition maps. The only fact
that matters is the initial state purity ($\varrho =
\ket{\psi}\bra{\psi}$) and the overlap with the transposed state
($|\ip{\psi}{\overline{\psi}}|^2$).

In the case of combined local and global noises
(section~\ref{subsection-bipartite-depolarizing}), robust states
can be either entangled (for dominating global noise) or
factorized (for dominating local noise). In fact, domination of
the global depolarizing noise corresponds to large values of
$\gamma > 0$ and small values of $\alpha,\beta \leqslant 0$, when
$\alpha^2 m +\beta^2 n + 2\gamma(\alpha+\beta) \leqslant 0$ and
the maximally entangled states exhibits higher output purity than
factorized states. On the other hand, for local depolarizing
channels $\Phi_{q_1} \otimes \Phi_{q_2}$ we have $\alpha = \frac{n
q_2}{1-q_2}$, $\beta = \frac{m q_1}{1-q_1}$, and $\gamma =
\frac{mn q_1 q_2}{(1-q_1)(1-q_2)}$; a direct calculation yields
$\alpha^2 m +\beta^2 n + 2\gamma(\alpha+\beta) > 0$ if $\Phi_{q_1}
\otimes \Phi_{q_2}$ is positive, i.e. factorized initial states
result in a larger output purity when local noises dominate.

\section{Conclusions}
\label{section-conclusions}

In this paper, we have revised the notion of absolutely separable
states with respect to bipartitions and multipartitions. In
particular, we have found an interesting example of the
three-qubit state, which is absolutely separable with respect to
partition $2|4$, and consequently is separable with respect to any
bipartition $\mathcal{H}_8 = \mathcal{H}_2 \otimes \mathcal{H}_4$,
yet is not separable with respect to tripartition $\mathcal{H}_2
\otimes \mathcal{H}_2 \otimes \mathcal{H}_2$.

We have introduced the class of absolutely separating maps and
explored their basic properties. This class is a subset of
entanglement annihilating maps. Even in the case of local maps, a
set of absolutely separating channels is not a subset of
entanglement breaking channels. In general, a map can be positive
and absolutely separating even if it is not completely positive.
We have shown that one-sided channels cannot be absolutely
separating, i.e. entanglement of the output state can always be
recovered by a proper choice of the input state and the unitary
operation applied afterwards. Even if the maps $\Phi_1$ and
$\Phi_2$ are absolutely separating with respect to partitions
$m_1|n_1$ and $m_2|n_2$, the tensor product $\Phi_1 \otimes
\Phi_2$ can still be not absolutely separating with respect to
partition $m_1 m_2 | n_1 n_2$. Global depolarizing maps are
absolutely separating if and only if they are entanglement
annihilating.

We have also analyzed $N$-tensor-stable absolutely separating maps
$\Phi$, whose tensor power $\Phi^{\otimes N}$ is absolutely
separating with respect to any valid partition. The greater $N$,
the closer an $N$-tensor-stable absolutely separating map
$\Phi:\mathcal{S}(\mathcal{H}_d) \mapsto
\mathcal{S}(\mathcal{H}_d)$ to the tracing map ${\rm Tr}[\varrho]
= {\rm tr}[\varrho] \frac{1}{d} I_d$. In fact, the tracing map is
the only map that is $N$-tensor-stable absolutely separating for
all $N$.

Particular characterization of absolutely separating property is
fulfilled for specific families of local and global maps. We have
fully determined parameters of two-qubit local depolarizing
absolutely separating maps $\text{PAS}(2|2)$ and provided
sufficient conditions for local Pauli maps. The factorized pure
states are shown to be the most robust to the loss of property
being not absolutely separable under the action of local noises.
Global noises are studied by examples of generalized Pauli
channels and combination of tracing map, transposition, and
identity transformation. Finally, the combination of local and
global noises is studied by an example of so-called bipartite
depolarizing maps. Robust states are shown to be either entangled
or factorized depending on the prevailing noise component: global
or local.

\ack The authors thank M\'{a}rio Ziman for fruitful discussions.
The authors are grateful to anonymous referees, whose valuable
comments helped to strengthen the results (in particular,
proposition~\ref{proposition-concatenation}) and stimulated the
analysis on $N$-tensor-stable absolutely separating maps
(section~\ref{section-n-tensor-stable-as}). The study is supported
by Russian Science Foundation under project No. 16-11-00084 and
performed in Moscow Institute of Physics and Technology.

\end{document}